\documentclass[11pt]{article}
\usepackage{color}
\usepackage[pdftex, 
letterpaper=true, colorlinks=true,
pdfpagemode=none, urlcolor=blue, linkcolor=blue, citecolor=blue,
pdfstartview=FitH] {hyperref}
\newcommand{\rg}[1]{\noindent{\textcolor{blue}{\textbf{\#\#\# RG:} \textsf{#1} \#\#\#}}}

\newcommand{\as}[1]{\noindent{\textcolor{blue}{\textbf{\#\#\# AS:} \textsf{#1} \#\#\#}}}
\usepackage{bbm}
\usepackage{amsmath,amsbsy,amsfonts,amssymb,amsthm,dsfont,fullpage,units}
\usepackage{cleveref}
\usepackage{paralist}
\DeclareMathOperator{\diam}{diam}

\DeclareMathOperator{\tr}{Tr}

\let\es=\emptyset

\newcommand{\lasserreii}[2]{\mathrm{Lasserre}_{#1}(#2)}

\newcommand\E{\mathbb{E}}
\newcommand\R{\mathbb{R}}

\newcommand\norm[1]{\left\|#1\right\|_2}

\newcommand\sse{\mbox{SSE}}
\newcommand\sdp{\mbox{SDP}}
\newcommand\opt{\mbox{OPT}}

\newcommand\xmat{X}

\newcommand\xvec{x}

\newcommand\eps{\epsilon}

\newtheorem{theorem}{Theorem}[section]
\newtheorem*{namedtheorem}{\theoremname}
\newcommand{\theoremname}{testing}

\newtheorem{lemma}[theorem]{Lemma}

\newtheorem{clm}[theorem]{Claim}
\Crefname{clm}{Claim}{Claims}
\newtheorem{proposition}[theorem]{Proposition}

\newtheorem{corollary}[theorem]{Corollary}

\theoremstyle{definition}
\newtheorem{definition}[theorem]{Definition}

\title{Towards a better approximation for {\sc sparsest cut}?}
 
\author{Sanjeev Arora\footnote{Princeton University, Computer Science
    Department and Center for Computational Intractability. Email: {\tt
      arora@cs.princeton.edu}}
  \and 
  Rong Ge\footnote{Princeton University, Computer Science Department
    and Center for Computational Intractability. Email: {\tt
      rongge@cs.princeton.edu}}
  \and
  Ali Kemal Sinop\footnote{Princeton University, Computer Science
    Department and Center for Computational Intractability. Email:
    {\tt asinop@cs.cmu.edu}} }
\date{\today}

\begin{document}
\maketitle

\begin{abstract}
  We give a new $(1+\epsilon)$-approximation for {\sc sparsest cut}
  problem on graphs where small sets expand significantly more than
  the sparsest cut (sets of size $n/r$ expand by a factor $\sqrt{\log
    n\log r}$ bigger, for some small $r$; this condition holds for
  many natural graph families). We give two different algorithms. One
  involves Guruswami-Sinop rounding on the level-$r$ Lasserre
  relaxation. The other is combinatorial and involves a new notion
  called {\em Small Set Expander Flows} (inspired by the {\em expander
    flows} of \cite{ARV})
  which we show exists in the input graph. Both algorithms run in time
  $2^{O(r)} \mathrm{poly}(n)$.

  We also show similar approximation algorithms in graphs with genus
  $g$ with an analogous local expansion condition.

  This is the first algorithm we know of that achieves
  $(1+\epsilon)$-approximation on such general family of graphs.
\end{abstract}
\section{Introduction}

This paper concerns a new and promising analysis of
Lasserre~\cite{Las02}/Parrilo~\cite{parrilo03} SDP relaxations for the
(uniform) {\sc sparsest cut} problem, which gives
$(1+\epsilon)$-approximation on several natural families of graphs.
Note that Lasserre/Parillo relaxations subsume all relaxations for the
problem that were previously analysed: the spectral technique of
Alon-Cheeger~\cite{am85}, the LP relaxation of
Leighton-Rao~\cite{lr99} with approximation ratio $O(\log n)$, and the
SDP with triangle inequality of Arora, Rao, Vazirani~\cite{ARV} with
approximation ratio $O(\sqrt{\log n})$.
The approximation ratio of $O(\sqrt{\log n})$ has proven resistant to improvement in
almost a decade  (and there is some evidence the ratio may be tight for
the ARV relaxation; see Lee-Sidiropoulos\cite{ls11}). For a few families of graphs such as
graphs of constant genus, an $O(1)$-approximation is known. 

Recently, there has been increasing optimism among experts that
Lasserre~\cite{Las02}/Parrilo~\cite{parrilo03} relaxations ---which
are actually a hierarchy of increasingly tighter relaxations whose
$r$th level can be solved in $n^{O(r)}$ time--- may provide better
approximation algorithms for {\sc sparsest cut} as well as other
problems such as {\sc max cut} and {\sc unique games}, and possibly
even refute Khot's unique games conjecture. For instance Barak,
Raghavendra, and Steurer~\cite{brs11}, relying on the earlier
subexponential algorithm of Arora, Barak, Steurer~\cite{ABS}, showed
that Lasserre relaxations can be used to design subexponential
algorithms for the {\sc unique games} problem. Independently,
Guruswami and Sinop~\cite{gs11-qip} gave another rounding that looks
quite different but yielded very similar results.  Subsequently, Barak
et al.~\cite{bhksz12} showed that Lasserre relaxations can easily
dispose of families of {\sc unique games} instances that seemed
``difficult'' for simpler SDP relaxations: many families of instances
can be solved near-optimally in 4-8 rounds! This result was
subsequently extended by O'Donnell~and~Zhou~\cite{oz13} to
``difficult'' families of graphs from~\cite{dksv06} which are
integrality gaps for uniform sparsest cut and balanced separator. Of
course, it is unclear whether this demonstrates the power of these
relaxations, or merely the limitations of our current lowerbound
approaches. Nevertheless, the rise in researchers' hopes for better
algorithms is palpable.

But the stumbling blocks in this quest are also quite clear.  First,
known ideas for analysing Lasserre relaxations generally require some
condition on the $r^{th}$ eigenvalue of the Laplacian for some small
$r$, whereupon some $f(r,\epsilon)$ levels of Lasserre are shown to
suffice for $(1+\epsilon)$-approximation.
Unfortunately, many real-life graphs (eg, even the 2D-grid) do not
satisfy this eigenvalue condition so new ideas seem needed.

Another stumbling block has been the inability to relate these new
rounding algorithms for Lasserre relaxations to existing SDP rounding
algorithms such as Goemans-Williamson and ARV. Since Lasserre
relaxations greatly generalize normal SDP relaxations, one would like
general purpose rounding algorithms which for small $r$ reduce to
earlier rounding algorithms. A concrete question is: does the
Guruswami-Sinop rounding algorithm always give an approximation ratio
as good as the ARV $\sqrt{\log n}$ for {\sc sparsest cut} once $r$ is
sufficiently large? This is still unclear.

The current paper makes some progress on these stumbling blocks.  We
show that the GS rounding algorithm achieves
$(1+\epsilon)$-approximation for {\sc sparsest cut} on an interesting
family graphs that are {\em not} small set expanders and may not have
large $r$th eigenvalue.
If $\phi_{local}$ denotes the minimum sparsity of sets of size $n/r$,
and $\phi_{sparsest}$ the minimum sparsity among {\em all} sets, then
we require $\phi_{local}/\phi_{sparsest} \gg \sqrt{\log n \log r}$.
Note that $\phi_{local}$ is often larger than $\phi_{sparsest}$ in
natural families of graphs. For example, in normalized $d$-dimensional
$n^{1/d}\times \ldots \times n^{1/d}$-grid graphs, $\phi_{sparsest}
\le \frac{1}{d n^{1/d}}$, $\phi_{local} \gg \frac{1}{d}
\left(\frac{r}{n}\right)^{1/d}$ whereas $\lambda_r \ll \frac{1}{d}
\left(\frac{r}{n}\right)^{2/d}$.
%
Note that when the condition is not met, a simple modification of our
algorithm returns a subset of size $n/r$ that has sparsity
$\sqrt{\log n \log r}$ times $\phi_{sparsest}$. Thus setting $r=O(1)$
one recovers the ARV bound ---though the analysis of this case also
uses ARV\footnote{We also know how to achieve qualitatively similar
  results as our main result using BRS rounding + ARV ideas applied to
  Lasserre solutions at the expense of stricter requirements on
  small set expansion. However, that method seems unable to give
  better than $O(1)$-approximation, whereas GS rounding is able to
  give $(1+\epsilon)$.}.
\paragraph{Comparison with existing work.}
As mentioned, earlier analyses of Lasserre relaxations require a
lowerbound on the $r^{th}$ eigenvalue of the graph: the tightest such
result from~\cite{gs11-exp} requires $\lambda_r > \phi_{sparsest}$.
Efforts to get around such limitations have focused on understanding
structure of graphs which {\em do not} satisfy the eigenvalue
condition: an example is the so-called {\em high order Cheeger
  inequality} of~\cite{ABS} (improved by Louis et al~\cite{lrtb12} and
Lee et al.~\cite{lgt12}) according to which --roughly speaking---a
graph with many eigenvalues close to $o(1)$ have a small nonexpanding
set. In other words, the graphs are not {\em Small-Set
  Expanders}\footnote{In fact, the unexpected appearance of Small Set
  Expansion (SSE) in this setting is believed to not be a fluke. It
  appears in the SSE conjecture of Raghavendra and
  Steurer~\cite{rs10-sse} (known to imply the UGC), their ``Unique
  games with SSE'' conjecture, as well as in the known subexponential
  algorithms for {\sc unique game}.  Furthermore, attempts to
  construct difficult examples for known SDP-based algorithms also end
  up using graphs (such as the noisy hypercube) which are small set
  expanders.}. However, there is an inherent Cheeger-like gap ($\phi$
vs $\sqrt{\phi}$) between eigenvalues and expansion that seems to
limit the possible improvements. Our algorithms work even without a
bound on the $r^{th}$ eigenvalue; they only need bounds on expansion
(The $d$-dimensional grids are good examples.)  Furthermore, they
yield $(1+\epsilon)$-approximation, which in context of {\sc sparsest
  cut} seems quite surprising.

Subsequent to our work and inspired by it,
Gharan~and~Trevisan~\cite{gt13} have shown how to obtain factor
$O(\sqrt{\log k})$ approximation from the basic ARV~relaxation for the
sparsest cut problem under local expansion or spectral conditions.
\paragraph{Better algorithms for bounded genus graphs.}
Recall that for genus $g$ graphs there are known $O(\log
g)$-approximation algorithms for {\sc sparsest cut}~\cite{ls10}. We can
show that GS'13 rounding gives a $(1+\epsilon)$-approximation if
$\phi_{local} \ge \Omega(\frac{\log g}{\epsilon^2}) \phi_{sparsest}.$
Thus for the $2D$-grid, it implies that $O(1/\epsilon^4)$ rounds of
Lasserre yield a $(1+\epsilon)$-approximation. Again, when the local
expansion condition is not satisfied our algorithm finds a witnessing
small set, allowing us to recover the existing $O(\log g)$
approximation for the general case.

%
\paragraph{Combinatorial algorithm.} In addition to the above
Lasserre-based algorithm, we also give a new combinatorial algorithm
with similar (but somewhat weaker) guarantees. This algorithm is
inspired by the primal-dual algorithms for {\sc sparsest cut} stemming
from the {\em expander flows} notion of ARV
(see~\cite{AHK1,AK,Sherman}). We introduce a new notion called {\em
  small set expander flows}: a multicommodity flow whose demand graph
is an expander on small sets. Let a {\em $(r, d, \beta)$-flow} be an
undirected multicommodity flow in which $d$ units of flow is incident
to each node, and the demand graph has expansion $\beta$ on sets of
size at most $n/r$ (in other words, the amount of flow leaving the set
$S$ is $d\beta |S|$). We show that in every graph there is an SSE flow
with $d =\Omega(\phi_{local}\sqrt{\log r}/\sqrt{\log n})$, $\beta =
\Omega((\log r)^{-2})$, and this flow ---or something close to
it---can be found in polynomial time. Using such flows one can ---with
some more work---compute a $(1+\epsilon)$-approximation to {\sc
  sparsest cut} as above.

Note that the expander flow idea of ARV was motivated by the
observation that expander flows consist of a family of dual solutions
to the SDP. We suspect that something analogous holds for SSE flows
and the Lasserre relaxation but are unable to prove this formally.
However, we can informally show a connection as follows: if a graph
has a $(r, d, \beta)$-flow where
$$d\beta^2/\log r\gg \mbox{value of $O(r)$-rounds
  of Lasserre relaxation}$$ then the integrality gap of the Lasserre
relaxation is at most $(1+o(1))$.  Thus the existence of expander
flows is another reason ---besides the more direct rounding approach
mentioned earlier--- why Lasserre relaxations are near-optimal when
$\phi_{local}/\phi_{global} \gg \sqrt{\log n \log r}$.

\paragraph{Applications to semirandom models}
%
Recently there has been interest in solving sparsest cut on semirandom
models of graphs~\cite{mmv}. In these graphs one starts with a planted
sparse cut in a random graph or expander, and then an adversary is
allowed to change some edges.  Our work provides a new algorithm for
one such model, planted combinatorial expander on regular
graphs. However our results for this model are not directly comparable
to the ones in~\cite{mmv}. Our result is presented
in~\Cref{sec:semirandom}.
%
%


\section{Preliminaries and Background}
\label{sec:prelims}


\subsection{Expansion and Graph Laplacian}

Let $G = (V,E)$ be an undirected graph with edge capacities $c_e \ge
0$ for all $e\in E$. For simplicity 
we 
assume that the input graph is regular with (normalized) degree 1,
that is, for all vertices $i\in V$ $\sum_{j} c_{(i,j)} = 1$ (our
results in \Cref{sec:orthogonal,sec:planar-graphs} can also
be applied to irregular graphs). We always use $n$ to denote the
number of vertices in $G$.

The {\em expansion} of a set is defined as $\Phi(S) =
\frac{E(S,V\setminus S)}{\min\{|S|,n-|S| \} }$, where $E(A,B) =
\sum_{i\in A,j\in B} c_{(i,j)}$. The sparsity of a set $\phi(S)$ is
defined as $\frac{n\cdot E(S,V\setminus S)}{|S|\cdot (n-|S|) }$. There
are several problems related to sparsity of cuts:
\begin{itemize}
\item The {\em sparsest cut} of the graph is a set $S$ that minimizes
  the sparsity $\Phi(S) = \frac{E(S,V\backslash S)}{|S|\cdot (n-|S|)
    \}}$. We use $\Phi_{sparsest}$ to denote its expansion and
  $\phi_{sparsest}$ to denote its sparsity.

\item The {\em edge expansion} of a graph is a set $S$ that minimizes
  the expansion $\alpha(S)$. We use $\Phi_{global}$ to denote its
  expansion. Notice that since we are working with regular graphs,
  this is also equivalent to the {\em graph conductance} problem.

\item The {\em $c$-balanced separator} of a graph is a set $S$ that
  minimizes the expansion $\Phi(S)$ among all sets of size at least
  $cn$. We use $\Phi_{c\mbox{-balanced}}$ to denote its expansion.
\end{itemize}

While all these problems are closely related (for example, sparsest
cut and edge expansion are equivalent up to a factor of 2), we
carefully differentiate between them in this paper because we are
looking for $1+\epsilon$ approximation algorithms.

We are also interested in the expansion of small sets: let $\Phi_r(G)$
be the smallest expansion of a set of size at most $n/r$ and
$\phi_r(G)$ be the smallest sparsity of a set of size at most
$n/r$. Sometimes when $r$ is fixed (or understood) we drop $r$ and use
$\Phi_{local}$ and $\phi_{local}$ instead\footnote{$\Phi_{local}$ and
  $\phi_{local}$ usually denote the optimal expansion and sparsity of
  sets of size at most $O(n/r)$}.

Notice that the requirement of our algorithms will have the form
$\phi_{local}/\phi_{global} \gg f(n,r)$\footnote{$f \gg g$ means $f \ge C g$ for some large universal constant $C$}. Since sparsity $\phi$ and
expansion $\Phi$ are within a factor of 2 ($\Phi(S) \le \phi(S)\le
2\Phi(S)$), in such requirements the ratios
$\phi_{local}/\phi_{global}$ and $\Phi_{local}/\Phi_{global}$ can be
interchanged

The adjacency matrix $A$ of the graph $G$ is a matrix whose $(i,j)$-th
entry is equal to $c_{(i,j)}$. If $d_i = \sum_{(i,j)\in E} c_{(i,j)}$
denotes the degree of $i$-th vertex with $D$ being the diagonal matrix
of degrees, then the Laplacian of the graph $G$ is defined as $L =
D - A$ (for regular graph this is just $I-A$). The
normalized Laplacian of the graph is defined as $\mathcal{L} =
D^{-1/2} L D^{-1/2}$.

Graph Laplacians are closely related to the expansion of sets. In
particular, the Rayleigh Quotient of a vector $x$, $R(x) =
\frac{x^TLx}{x^Tx}$ is exactly equal to the sparsity of a set $S$
when $x$ is the indicator vector of $S$ (and $S$ has size at most $1/2$).

We will denote by $\phi_{SDP}$ the optimum value of the Lasserre
relaxation for {\sc sparsest cut}. The number of levels in the
Lasserre hierarchy will be implicit in the context.
\subsection{Lasserre Relaxation and GS Rounding}
\label{subsec:lasserreprelim}
We will show sufficient conditions under which $r$ rounds of Lasserre
Hierarchy relaxation can be rounded to $(1+\epsilon)$-approximation
for sparsest cut and related problems.  In particular, we will show
that the particular rounding algorithm from~\cite{gs11-exp} outputs
such an approximation. (See~\Cref{sec:overview-sdp} for details on the
Lasserre relaxation and the GS rounding algorithm.)
%
%

In general working with Lasserre relaxations involves tedious notation
involving subsets of variables and assignments to them. Luckily all
that has been handled in~\cite{gs11-exp}, leaving us to work with the
relatively clean (standard) SDP notation.

For the sake of simplicity, we will focus on the uniform sparsest cut
problem on regular graphs. Other variants, such as edge expansion, can
easily be handled by changing the objective function. Let
$[\xvec_u]_{u\in V}$ be the vectors corresponding to each node in $G$
obtained as a solution for $r$-rounds of Lasserre Hierarchy
relaxation. In particular, $\xvec_u$'s minimize the following ratio:
\[
\phi_{SDP} \triangleq
\frac{\sum_{u<v} C_{uv} \|\xvec_u - \xvec_v\|^2}{\frac{1}{n}
  \sum_{u<v} \|\xvec_u-\xvec_v\|^2} \le \phi_{sparsest}.
\]
The denominator, whose value we will denote by $\nu$, can also be
written as:
\[
\underbrace{\frac{1}{n} \sum_{u<v} \|\xvec_u-\xvec_v\|^2}_{\triangleq
  \nu} = \sum_u \Big\| \xvec_u - \frac{1}{n}\sum_v \xvec_v\Big\|^2.
\]
We will shift each vector $\xvec_u$ by the mean:
\begin{equation}
  \label{eq:sdp-demand-vectors}
  \xmat_u \triangleq \xvec_u - \frac{1}{n}\sum_v \xvec_v, 
\end{equation}
so that $\sum_u \xmat_u = 0$.
Note:
\[
  \|\xmat_u\|^2 \le 1.
\]
We use $\xmat=[\xmat_u]$ to denote the matrix whose columns correspond
to the vectors $\xmat_u$.
Since $\xmat_u - \xmat_v = \xvec_u - \xvec_v$, $\xmat
\in \ell^2_2$ (i.e. columns of matrix $\xmat$ satisfy the triangle
inequality) and:
\[
\sum_{uv} C_{uv} \|\xmat_u - \xmat_v\|^2
= \sum_{uv} C_{uv}\|\xvec_u- \xvec_v\|^2 \le \phi_{SDP}
\frac{1}{n} \sum_{u<v} \|\xvec_u - \xvec_v\|^2
= \phi_{SDP} \|\xmat\|^2_F
\] where last identity follows from the fact that $\sum_u \xmat_u = 0$.
Using $\xmat$, we can re-state Theorem 3.1 from~\cite{gs11-exp} in the 
following way:
 \begin{theorem}[Theorem 3.1 from~\cite{gs11-exp}]\label{thm:GS13}
   \label{thm:gs-round-or-fail}
   If there exists a subset $S\in \binom{V}{r}$ with
   \begin{equation}
     \label{eq:sc-large-proj-dist}
      \|\xmat_S^\perp \xmat\|^2_F=
     \sum_{u} \|\xmat_S^\perp \xmat_u\|^2 \le \gamma \|\xmat\|^2_F,    
   \end{equation}
   then the rounding algorithm from~\cite{gs11-exp} outputs a set $T$ such that:
   \[
     \phi_G(T) \le \frac{\phi_{SDP}}{1-\gamma} .
   \]
   Here $\xmat_S$ is the projection matrix onto the span of the
   submatrix indexed by $S$ and $\xmat_S^\perp$ is the projection matrix onto the orthogonal
   complement of $\xmat_S$'s column span. 

   Furthermore, the SDP solver and rounding procedure can be
   implemented in time $2^{O(r)} \mathrm{poly}(n)$
   using~\cite{gs12-fast}.
\end{theorem}


\section{Proof via orthogonal separators}
%
%
\label{sec:orthogonal}
\Cref{thm:GS13} implies that for $(1+\epsilon)$-approximation it
suffices to show the {\em existence} of a small subset $S$ of vertices
such that the relative distance of all other vertices to the span of
$\xmat_S$ is smaller than any small constant.

\begin{theorem}[Main] \label{thm:main1} For every $\epsilon >0$ there
  is a constant $C =C(\epsilon)$ such that the following is true.
  When all subsets of at most $2n/r$ vertices have sparsity
  $\phi_{local} \ge C \phi_{SDP} \sqrt{\log n \log r}$ in the graph,
  there exists a set $S$ of $r$ vertices such that $\|X_S^\perp
  X\|_F^2 \le \epsilon \|X\|_F^2$.  (Here $\phi_{SDP} \le
  \phi_{sparsest}$ is the value of the Lasserre relaxation for $r+3$
  rounds and $\xmat$'s are the corresponding vectors
  from~\cref{eq:sdp-demand-vectors}.)
\end{theorem}
%
%
This existence result will be proven using the {\em orthogonal
  separators}~\cite{cmm06} but with the modifications of Bansal et
al.\cite{bansalsse11}, which, not surprisingly, were also developed in
context of algorithms for small set expansion.  (We know how to give a
more direct proof without using orthogonal separators but it brings in
an additional factor of $\log r$ in the local expansion condition.)

\begin{definition}[Orthogonal Separator] Let $X$ be an $\ell_2^2$
  space. A distribution over subsets of $X$ is called an
  $m$-orthogonal separator with {\em distortion} $D$, {\em probability
    scale} $\alpha > 0$ and {\em separation threshold} $\beta < 1$ if
  the following conditions hold for $S\subset X$ chosen according to
  this distribution.

  \begin{enumerate}
  \item For all $X_u\in X$, $\Pr[X_u\in S] = \alpha \|X_u\|^2$.
  \item For all $X_u,X_v\in X$ with $\|X_u-X_v\|^2 \ge \beta
    \min\{\|X_u\|^2, \|X_v\|^2\}$,
    $$
    \Pr[X_u\in S\mbox{ and }X_v\in S] \le \frac{\min\{\Pr[X_u\in S],
      \Pr[X_v\in S]\}}{m}.
    $$
  \item For all $X_u,X_v\in X$, $\Pr[I_S(X_u)\ne I_S(X_v)] \le \alpha
    D\cdot \|X_u - X_v\|^2$, where $I_S$ is the indicator function of
    $S$.
  \end{enumerate}
\end{definition}
Bansal et al.~\cite{bansalsse11} showed the existence of such
separators (in the process also giving an efficient algorithm to construct
them).
\begin{lemma}[\cite{bansalsse11}]
  For all $\beta < 1$ there exists an $m$-orthogonal separator with
  distortion $D = O\left(\sqrt{\frac{\log |X|\log m}{\beta}}\right)$.
\end{lemma}
The dependency on $\beta$ follows from calculations in Lemma 4.9 in
\cite{cmm06}. From the explanation of the above Lemma in
\cite{bansalsse11}, we know $\gamma = \sqrt{\beta}/8$, so the exponent
in Lemma 4.9 in \cite{cmm06} is $1/(1-\gamma^2)-1 = O(\beta)$, and we
want $(\log m'/m')^{O(\beta)}$ to be smaller than $1/m$.  Setting $m'
= m^{O(1/\beta)}$ suffices. Then the distortion is $O(\sqrt{\log
  |X|\log m'}) = O\left(\sqrt{\frac{\log |X|\log m}{\beta}}\right)$.
%

Now we show the following, which immediately implies \Cref{thm:main1}.

\begin{theorem} \label{thm:mainorth} For any $\delta > 0$, $0.25 >
  \beta > 0$, let $m = 10r^2/\delta$.  Let $D$ denote the best
  distortion possible for an $m$-orthogonal separator with separation
  $\beta$.  If $X$ is any set of vectors in $\ell_2^2$, one for each
  vertex in the graph, and the minimum expansion $\phi_{local}$ among
  subsets of at most $2n/r$ vertices satisfies $\phi_{local} \ge
  O(\phi_{SDP} D/ \delta)$, then there exist $r$ points $S$ in $X$ such
  that $\|X_S^\perp X\|_F^2 \le O(\delta+\beta) \|X\|_F^2$.
\end{theorem}
The actual construction of orthogonal separators from~\cite{cmm06}
requires the origin to be inside the vector set. To achieve this, we
will translate all vectors in the same direction:
\begin{proposition} \label{thm:new-x-with-origin}
  If~\Cref{thm:gs-round-or-fail} fails, then there exists a set of
  vectors $\xmat \in \ell_2^2$ with $0 \in \xmat$.
\end{proposition}
\begin{proof}
  Given the vectors $[\xmat_u]_u$ found
  by~\Cref{thm:gs-round-or-fail}, we know that $ \sum_u \|\xmat_u\|^2
  = \frac{1}{2}\mathbb{E}_{u} \sum_{v} \|\xmat_u - \xmat_v\|^2.  $
  Hence there exists some $t$ for which $ \sum_{u} \|\xmat_u -
  \xmat_t\|^2 \le 2\sum_u \|\xmat_u\|^2.  $ After having fixed such
  $t$, we define our new vectors as $ \xmat'_u \gets \xmat_u -
  \xmat_t.  $ It is easy to see that $\xmat'\in \ell^2_2$, $0 \in
  \xmat'$ and for every subset of size $r-1$,
  \cref{eq:sc-large-proj-dist} is satisfied (except $\epsilon$ becomes
  $\epsilon/2$, which only changes the constants in $O$ notation).
\end{proof}
%
%
We start by showing that most sets in the support of the orthogonal separator should be
large. 
\begin{lemma}
\label{lem:smallsets}
  If $\phi_{local} \ge 2\phi_{SDP} D/ \delta$ as in the hypothesis of
  \Cref{thm:mainorth}, and $S$ is chosen according to the orthogonal
  separator, then $ \E[|S|\cdot I_{|S| \le 2n/r}] \le \delta \E[|S|],
  $ where $I_{|S|\le 2n/r}$ is the indicator variable for the event
  ``$|S| \le 2n/r$.''
\end{lemma}
\begin{proof}
On one hand we know
$$
\E[\mbox{number of edges cut}]
\ge \E[|S|\cdot I_{|S| \le 2n/r}]\cdot \Phi_{local} \ge \E[|S|\cdot I_{|S| \le 2n/r}]\cdot \phi_{local}/2.
$$
On the other hand by condition 3 in the definition,
$$
\E[\mbox{number of edges cut}]
\le\alpha D  \sum C_{uv} \|X_u-X_v\|^2.
$$
Since $\sum C_{uv} \|X_u-X_v\|^2 \le \phi_{SDP} \sum \|X_u\|^2  = \phi_{SDP} \E[|S|]/\alpha$, we know
$$
\E[|S|\cdot I_{|S| \le 2n/r}] \le \frac{1}{\phi_{local}}\alpha D  \sum C_{uv} \|X_u-X_v\|^2
\le \delta \E[|S|].\qedhere
$$
\end{proof}
Now we state a corollary but first we need this definition.
%
\begin{definition}[volume]
  The {\em volume} of a subset $X'\subset X$ is $$vol(X') =
  \frac{\sum_{X_u\in X'} \|X_u\|^2}{\sum_{X_u\in X} \|X_u\|^2}$$.
\end{definition}
\begin{corollary} \label{cor:S'} There exists a subset $X'\subset X$
  with volume at least $1-2\delta$, such that the following is
  true. Let $S'$ be a set picked probabilistically by first picking
  $S$ randomly according to the separator and letting $S'$ be $S$ if
  $|S| \ge 2n/r$, and the empty set otherwise.  Then we have:
  \begin{enumerate}
  \item For all $X_u\in X'$ we have $\Pr[X_u\in S'] \ge \alpha
    \|X_u\|^2/2$.
  \item For all $X_u,X_v\in X$ with $\|X_u-X_v\|^2 \ge \beta
    \min\{\|X_u\|^2, \|X_v\|^2\}$,
    $$
    \Pr[X_u\in S'\mbox{ and }X_v\in S'] \le \frac{\min\{\Pr[X_u\in S],
      \Pr[X_v\in S]\}}{m}.
    $$
  \end{enumerate}
\end{corollary}
\begin{proof}(Sketch) The first condition is by Markov. The second
  condition holds because the $S'$ is always a subset of $S$, so the
  probability of LHS only decreases.
\end{proof}
Now we are ready to prove \Cref{thm:mainorth}.
\begin{proof}(\Cref{thm:mainorth}) 
We give an algorithm that shows iteratively picks $r$ points such that
most of the volume in $X'$ lies close to them.
%
\begin{enumerate} \itemsep 0pt
\item Initially none of the points are marked.
\item $i \gets 1$.
\item While there is still a point in $X'$ that is not marked:
  \begin{enumerate}
  \item Let $X_i$ be the point with largest norm among the unmarked points of $X'$.
\item \label{step:orth-sep-5}
  Pick a set $S_i$ ($|S_i| \ge 2n/r$) containing $X_i$,
  and containing at most $2n/m$ points that have distance  more than $\beta \|X_i\|^2$
  from $X_i$. (Such a set exists as shown below.)
\item Mark all points in $S_i$ as well as all points that have
  distance at most $2\beta \|X_i\|^2$ from $X_i$. Denote by $M_i$  the set of
  points that were previously unmarked and got  marked in this step.
\item Look over all $M_j$ for $j<i$ and if any points in them have distance at most $\beta
  \|X_i\|^2$ to $X_i$  then add them to $M_i$ as well.
\item $i \leftarrow i+1$.
\end{enumerate}
\end{enumerate}
First we show using the probabilistic method why we can always perform
step~\ref{step:orth-sep-5}.  Pick a random set $S'$ from the distribution of the
separator, conditioning on its containing $X_i$.  By the properties of
$S'$ we know if $\|X_i-X_v\|^2 \ge \beta \|X_i\|^2$, then the
conditional probability $\Pr[X_v\in S'|X_i\in S'] \le 2/m$. So the
expected number of points in $S'$ whose distance is at least
$\beta\|X_i\|^2$ from $X_i$ is at most $2n/m$, and in particular there
must be one set that satisfy the condition.
Then we need to show that this process terminates in $r$ steps. To do
so it suffices to show that each $S_i$ has at least $n/r$ points that
were not in any $S_j$ for $j<i$.  We know that $|S_i| > 2n/r$.  We
claim its intersection with any $S_j$ for $j<i$ is at most
$4n/m$. 
The reason is that $X_i$ was unmarked at the start of this phase,
which implies that that the balls of radius $\beta \|X_j\|^2$ and
$\beta \|X_i\|^2$ around $X_j$ and $X_j$ respectively must be disjoint
(note that $\|X_j\| > \|X_i\| $) and thus the only intersections among
$S_i, S_j$ are from points outside these balls, which we know to be at
most $4n/m$. Since $4nr/m <n/r$ (recall $m = 10r^2/\delta$), we have
conclude that each $S_i$ introduces at least $n/r$ new points, so the
process must terminate in $r$ steps.
Finally we bound the average distance of the other points from this
set $S =\{X_1, ..., X_t\} (t\le r)$, specifically, the quantity
$\frac{\|X_S^\perp X\|_F^2}{\|X\|_F^2}$ by $O(\delta +\beta)$.
All points outside $X'$ (the set in \Cref{cor:S'}) anyway have volume
at most $2\delta$, so their contribution is upperbounded by that.  To
bound the contribution of points in $X'$, consider how the sets $M_1,
..., M_t$ were picked. If $X_u$ is $\beta \|X_i\|^2$-close to $X_i$,
then $X_u$ is in $M_i$ (these sets are disjoint by the
construction). Otherwise $X_u$ belongs to $M_i$ where $i$ is the time
that $X_u$ gets marked.
All points in $M_i$ have norm at most $\|X_i\|^2$ since otherwise they
would have been picked instead of $X_i$. Also more than $n/r$ points
(in fact, $2n/r - 2n/m > n/r$) in $M_i$ are $\beta\|X_i\|^2$-close to
$X_i$, and at most $2n/m$ points are $2\beta\|X_i\|^2$-far from $X_i$,
so
$$
\sum_{X_u\in M_i}\|X_u\|^2 \ge (|M_i| - 2n/m) (1-2\beta)\|X_i\|^2 \ge
|M_i|\|X_i\|^2/3.
$$
On the other hand, after projection to the space orthogonal to $X_S$,
all but $2n/m$ points are smaller than $2\beta \|X_i\|^2$, therefore
after projection
$$
\sum_{X_u\in M_i} \|X_i^\perp X_u\|^2 \le (|M_i| - 2n/m) \cdot 2\beta
\|X_i\|^2 + 2n/m \cdot \|X_i\|^2 \le O(\beta+\delta)|M_i|\|X_i\|^2.
$$
Summing up the inequalities for all $M_i$'s we get the upperbound
$O(\beta +\delta)$ needed for the theorem.
\end{proof}
\paragraph{Algorithmic version.}
Since Bansal et al.~\cite{bansalsse11} give an efficient algorithm for
constructing orthogonal separators, the above proof immediately can be
made algorithmic.
\begin{corollary}\label{thm:orth-sep-algorithmic}
  There is an algorithm that given a weighted graph $G=(V,E)$ in which
  $\phi_{local} > \frac{O(\sqrt{\log n \log r/\epsilon})}{\eps^{3/2}}
  \phi_{sparsest}$ computes a $(1+\epsilon)$-approximation to {\sc
    sparsest cut} in time $2^{O(r)} \mathrm{poly}(n)$. Here
  $\phi_{local}$ is the minimum sparsity of sets of size at most $2
  n/r$.
  In fact the algorithm outputs one of the following. 
  \begin{enumerate}
  \item Either a subset with sparsity at most $(1+\eps)
    \phi_{\mathrm{SDP}}$,
  \item Or a subset of size at most $\frac{2 n}{r}$ with sparsity at
    most $\frac{O(\sqrt{\log n \log r/\epsilon})}{\eps^{3/2}}
    \phi_{SDP}$.
  \end{enumerate} Here $\phi_{\mathrm{SDP}}$ is the optimum value
  of~\cref{eq:sc-sdp-alg} for $r+3$ rounds.
\end{corollary}
\begin{proof}(Sketch) Consider the algorithm
  from~\Cref{thm:gs-round-or-fail}. If it outputs a partition, we are
  done. Otherwise, we apply the algorithm for constructing orthogonal
  separator in~\cite{bansalsse11} on the set of vectors as constructed
  in \Cref{thm:new-x-with-origin}. The above existence proof of the
  set $S$ fails for this set of vectors, therefore
  \Cref{lem:smallsets} fails, and there must be a small set in the
  orthogonal separator that has desired expansion.
  %
\end{proof}



 





\section{Bounded Genus Graphs}
\label{sec:planar-graphs}
In this section, we prove an analog of our result for graphs with orientable
genus $g$. 
The standard LP relaxation~\cite{lr99} for {\sc sparsest cut} on such
graphs has an integrality gap of $O(\log g)$~\cite{ls10}. For planar
graphs (when $g=0$), Park~and~Phillips~\cite{pp93} presented a weakly
polynomial time algorithm for the problem of edge expansion using
dynamic programming.
%

%
Here we show how to give a $(1+\epsilon)$-approximation when the graph
satisfies a certain local expansion condition. Note that this
expansion condition is true for instance in $O(1)$-dimensional grids
when $r=poly(1/\epsilon)$. 

\begin{theorem}\label{thm:bounded-genus-good-or-small}
  There is a polynomial-time algorithm that given a weighted graph $G$
with orientable genus $g$ in which
  $\phi_{local} > \frac{O(\log g)}{\eps^2} \phi_{sparsest}$ (where
  $\phi_{local}$ is the minimum sparsity of sets of size at most
  $n/r$) computes a $(1+\epsilon)$-approximation to {\sc sparsest cut}
  and similar problems in $2^{O(r)} \mathrm{poly}(n)$.

  In fact the algorithm outputs one of the following. 
  \begin{enumerate}
  \item Either a subset with sparsity at most $(1+\eps)
    \phi_{\mathrm{SDP}}$,
  \item Or a subset of size at most $\frac{n}{r}$ with sparsity at
    most $ \frac{O(\log g)}{\eps^{2}} \phi_{SDP}$.
  \end{enumerate} Here $\phi_{\mathrm{SDP}}$ is the optimum value
  of~\cref{eq:sc-sdp-alg} for $r+2$ rounds.
  %
\end{theorem} 
Before proving~\Cref{thm:bounded-genus-good-or-small}, let us first
recall the theory of random partitions of metric spaces, and its
specialization to graphs of bounded genus.  If $(V,d)$ is a metric
space then a {\em padded decomposition at scale $\Delta$} is a
distribution over partitions $P$ of $V$ where each block of $P$ has
diameter $\Delta$. Its {\em padding parameter} is the smallest $\beta
\ge 1$ such that the ball of radius $\Delta/\beta$ around a point has
a good chance of lying entirely in the block containing the point:
\begin{equation}
  \label{eq:padded-cond-1}
  \mathrm{Prob}_P[B_d(u, \Delta/\beta) \subseteq P(u)] \ge \frac{1}{8} 
  \text{ for all $u\in V$.}
\end{equation}
The {\em padding parameter} of a graph $G$ is the smallest $\beta$ such
that every semimetric formed by weighting the edges of $G$ has a
padded decomposition with padding parameter at most $\beta$.
The following theorems are known.
\begin{theorem} \label{thm:bounded-genus-padded}  
  \begin{enumerate}
  \item \cite{ls10} If $G$ has orientable genus $g$, then its padding
    parameter is $O(\log g)$.
  \item \cite{ft03} If $G$ has no $K_{p,p}$ minor, then
    its padding parameter is $O(p^2)$.
  \end{enumerate} 
\end{theorem}

Our main technical lemma is the following.
\begin{lemma} \label{lem:good-or-small-cut-for-padded} Given a graph
  $G=(V,E)$ and positive integer $r$, there exists an algorithm which
  runs in time $2^{O(r)} \mathrm{poly}(n)$ and outputs one of the
  following for any $\eps>0$:
  \begin{enumerate}
  \item Either a subset with sparsity at most $(1+\eps)
    \phi_{\mathrm{SDP}}$ where $\phi_{\mathrm{SDP}}$ is the optimum
    value of~\cref{eq:sc-sdp-alg},
  \item Or a subset of size at most $\frac{n}{r}$ with sparsity at
    most $\frac{O(\beta)}{\eps^2} \phi_{SDP}.$
  \end{enumerate}
\end{lemma}
\begin{proof}  
  The idea is to apply the algorithm
  from~\Cref{thm:gs-round-or-fail}. If it finds a cut of sparsity
  $(1+\eps) \phi_{\mathrm{SDP}}$, then we are done. Otherwise let
  $[\xmat_u]_u$ be the vectors output by it. We show how to use padded
  decompositions of the shortest-path semimetric given by distances
  $\|\xmat_u - \xmat_v\|^2$ and then produce a small nonexpanding set.

 Let $\nu$ denote the average squared length of these vectors,
  i.e. $\nu \triangleq \frac{1}{n} \sum_u \|\xmat_u\|^2$ so that $\nu
  =\mu (1-\mu)$.
  Choose $\Delta$ at least $ \frac{\eps}{2} \frac{\sum_u
    \|\xmat_u\|^2}{n}$. Take a padded decomposition at scale
  $\Delta$ and pick a random partition $P$ out of it.

\noindent{\sc Claim:} {\em The expected number of nodes that lie in
subsets of size less than $n/r$ in $P$ is at least $\frac{\eps}{2} \sum_{u}
  \|\xmat_u\|^2$.}
  %

 \noindent{\bf Proof} For each subset $S\in P$ with
  size $|S|\ge \frac{n}{r}$, if we choose an arbitrary $t \in S$,
  \cref{eq:sc-large-proj-dist} implies that:
  \begin{align*}
    \eps \sum_u \|\xmat_u\|^2 \le&
    \sum_{S\in P: |S| \ge \frac{n}{r} } \sum_{u\in S} \|\xmat_t -
    \xmat_u\|^2 
    + \sum_{T\in P: |T| < \frac{n}{r}} \sum_{u\in T} \|\xmat_u\|^2 \\
    \le & \sum_{S\in P: |S| \ge \frac{n}{r} } \Delta
    |S| + \sum_{T\in P: |T| < \frac{n}{r}} \sum_{u\in T} \|\xmat_u\|^2 
    \le \frac{\eps}{2} \sum_u \|\xmat_u\|^2  + 
    \sum_{T\in P: |T| < \frac{n}{r}} |T|. \\
    \frac{\eps}{2} \sum_u \|\xmat_u\|^2 \le & \sum_{T\in P: |T| < \frac{n}{r}} |T|.
  \end{align*}
  \qed
  
Now we choose a threshold $\tau \in [0, \Delta/\beta]$ uniformly at
  random. Then for each $T\in P$ with $|T| \le \frac{n}{r}$, let
  $\widehat{T} \subseteq T$ be the subset of nodes which are in the
  same partition block as the ball of radius $\tau$ around them.
We output such $\widehat{T}$ with minimum sparsity among all
  $T\in P$ with $|T|\le \frac{n}{r}$. Using standard arguments, we can
  prove that any pair of nodes $u$ and $v$ is separated with
  probability at most $\frac{\|\xmat_u -
    \xmat_v\|^2}{\Delta/\beta}$. This means the total expected
  capacity cut will be at most $\frac{\beta}{\Delta} \sum_{u<v} C_{uv}
  \|\xmat_u - \xmat_v\|^2$. Moreover~\cref{eq:padded-cond-1} implies
  that:
  \[
  \mathbb{E}_P \Big[ \sum_{T\in P: |T|\le n/r} |\widehat{T}| \Big] 
  \ge \frac{1}{8} \sum_{T\in P: |T|\le n/r} |{T}| \ge \frac{\eps}{16}
  \sum_u \|\xmat_u\|^2.
  \]
  Putting all together, we see that there exists some $T \in P$ with
  $|T|\le\frac{n}{r}$ such that
  \[
    \phi_G(T) 
  \le \frac{O(\beta)}{\eps^2} \phi_{SDP}. \tag*{\qedhere}
  \]
\end{proof}
%
%
%
Combining \Cref{lem:good-or-small-cut-for-padded} with the bounds
from~\Cref{thm:bounded-genus-padded} immediately implies
\Cref{thm:bounded-genus-good-or-small}.
%


\section{Small-set Expander Flows}


In \cite{ARV}, {\em expander flows} are used as approximate
certificate for expansion, which work for {\em all} values of
expansion. (By contrast, the eigenvalue or spectral bound of
Alon-Cheeger is most useful only for expansion close to $\Omega(1)$.)
This section concerns {\em small-set expander flows} (SSE flows) which
can be viewed as approximate certificates of the expansion of small
sets. An $(r,d,\beta)$-SSE flow is a multicommodity flows in which
small sets $S$ (ie sets of size at most $n/r$ for some small $r$) have
$\beta d|S|$ outgoing flow where $\beta$ is close to $\Omega(1)$. The
flow is undirected, and the amount of flow originates at every node is
at most
$d$. 
Since the flow resides in the host graph and $\beta d|S|$ leaves every
small small set $S$, an $(r,d,\beta)$-SSE flow is trivially a
certificate that small sets have edge expansion $\Omega(d\beta)$ in
the host graph.

Of particular interest here will be a surprising connection between SSE
flows and finding near-optimal {\sc sparsest cut}. In other words,
information about expansion of small sets can be leveraged into
knowledge about the expansion of all sets. We note that such a
leveraging was already shown in \cite{ABS} using spectral techniques,
but only when Small set expansion is $\Omega(1)$, roughly speaking (the
reason is that the proof is Cheeger-like).

We note that given a flow it seems difficult (as far as we know) to verify that it
is an SSE flow. Thus we will also be interested in a closely related
notion of {\em spectral SSE flow}, which by contrast is easily
recognized using eigenvalue computation. This is the one used in our algorithm.

\begin{definition}[Spectral SSE Flow]
  A {\em $(r, d,\lambda)$-spectral SSE flow} is a multicommodity flow
  whose vertices have degree between $d/2$ and $d$, and the $r^{th}$
  smallest eigenvalue of its Laplacian matrix is at least $d\lambda$.
\end{definition}

The relationship between the two types of flow rely upon the so-called
higher order Cheeger inequalities~\cite{lrtb12, lgt12}.

\begin{theorem}[Rough statement]
\label{thm:SSEconvert}
If the graph has an $(r,d,\beta)$ SSE flow then it also has an
$(2r,d,\Omega(\beta^2/\log r))$ spectral SSE flow. Conversely, if the
graph has an $(r,d,\lambda)$ spectral SSE flow then it has a weaker
version of $(r,d,\beta = \lambda)$ combinatorial SSE flow.
\end{theorem}

See \Cref{lem:combtospectral} and \Cref{lem:spectraltocomb} for more precise statements.

Now we describe how these results are useful.  First, just {\em
  existence} of SSE flows is enough to imply a low integrality gap for
the Lasserre relaxation. This is reminiscent of primal-dual frameworks
(e.g., expander flows being a family of dual solutions for the ARV SDP
relaxation and thus giving a lower bound on the optimum) but we don't
know how to make that formal yet.

\begin{theorem} \label{thm:SSE1} If a $(r,d,\lambda)$-spectral SSE
  flow exists in the graph for $d\lambda \gg
  \frac{1}{\epsilon}\phi_{sparsest}$, then the GS rounding algorithm
  computes a $(1+\epsilon)$-approximation to {\sc sparsest cut} when
  applied on the $O(r/\epsilon)$-level Lasserre
  solution. 
  In particular, the integrality gap of the Lasserre relaxation is at
  most $(1+\epsilon)$.
\end{theorem}

The other result is a more direct approximation algorithm that does
not use SDP hierarchies at all. Instead it uses a form of spectral
rounding (as in~\cite{ABS}) that produces a set with low symmetric
difference to the optimum sparsest cut, followed by the clever idea of
Andersen and Lang~\cite{al08} to purify this set into a bonafide cut
of low expansion.

\begin{theorem} \label{thm:SSE2} There is a $2^{O(r)} poly(n)$ time
  algorithm that given a graph and a $(r,d,\lambda)$-spectral SSE flow
  for $d\lambda \gg \frac{1}{\epsilon^2}\phi_{sparsest}$ outputs a cut of
  sparsity at most $(1+\epsilon)\phi_{sparsest}$.
\end{theorem}

The above two theorems become important only because of the following
two theorems which concern the {\em existence} of the
flow. 

\begin{theorem}\label{thm:SSEexist}
  If $d \ll \Phi_{local}\sqrt{\log r}/\sqrt{\log n}$ then the graph
  has a $(r,d, \Omega((\log r)^{-2}))$ SSE flow.
\end{theorem}

This theorem follows from \Cref{lemma:dualsolution} and
\Cref{lem:weaktosseflow}.


\begin{theorem}\label{thm:spectralSSEexist}
  If $d \ll \Phi_{local}\sqrt{\log r}/\sqrt{\log n}$ then the graph
  has a $(2r,d, \Omega((\log r)^{-5}))$ spectral SSE flow. Furthermore, a
  $(4r,d, \Omega((\log r)^{-5}))$ spectral SSE flow can be found in
  polynomial time.
\end{theorem}

This theorem follows \Cref{thm:SSEexist,thm:SSEconvert}
and \Cref{lem:findspectralsseflow}. The algorithm to find the
spectral SSE flow uses the fact that maximizing the sum of first $r$
eigenvalues of a matrix is a convex objective.

In fact, when $\Phi_{local}$ is small, we can actually find a small
set that does not expand well.

\begin{theorem}
  For any graph $G = (V,E)$ and any value $d$, there is a polynomial time algorithm
  that either finds a $(4r,d, \Omega((\log r)^{-5}))$ spectral SSE
  flow, or finds a set of size at most $100n/r$ that has expansion at
  most $O(d\sqrt{\log n}/\sqrt{\log r})$.
\end{theorem}

For more details see \Cref{lem:finddual}.

\subsection{Overview of proof of existence of SSE flows}

To keep the main paper relatively concise, we have move the proof of
existence to the appendix and give an overview here.

From a distance, the existence proof for SSE flows uses similar ideas
as the one for expander flows in \cite{ARV}: we write an exponential
size LP that is feasible iff the desired flow exists, and then reason
about the properties of dual solutions (using properties of flows,
cuts, and $\ell_2^2$ metrics) to show that the LP is feasible.

We write an LP that enforces each vertex has degree at most $d$ in the
flow, and for every set $S$ of size $n/3r$ to $n/r$, the amount of
outgoing flow is at least $\beta d|S|$, the precise LP can be found in
\Cref{subsec:LPformulation}.

The dual of this LP consists of a nonnegative weight $s_i$ for all
vertices and $w_e$ for each edge, and also a nonegative weight for every set of
size between $n/3r$ to $n/r$. We shall prove the following Lemma:

\begin{lemma}[imprecise]
  Given a valid dual solution with degree $d$ and $\beta$ parameter
$=\Theta((\log r)^{-2})$, there is an algorithm that finds a set of
  size at most $100n/r$ with expansion $O(d\sqrt{\log n}/\sqrt{\log r})$.
\end{lemma}

In order for the algorithm to run in polynomial time, we first need to
represent the LP dual concisely, and as stated above it involves a
nonegative weight on exponentially many cuts! As in ARV, this concise
representation is possible since 
a nonnegative weighting of cuts is an $\ell_1$ metric and 
the algorithm is only interested in the ``distance'' between two
vertices in this metric  (which is the measure of sets that contains one of the vertices but
not the other). The $\ell_1$ metric can be concisely represented by some
$\ell_2^2$ vectors; see \Cref{subsec:mapping}
and \Cref{lem:finddual}. 

The proof of the Lemma above uses the
``chaining'' idea from \cite{ARV}, but there are many
differences which we list here.

\begin{enumerate}[(a)]
\item The proof is handicapped since it is only allowed to use {\em
    local expansion} (i.e., expansion of sets of size at most
  $O(n/r)$), and this requires us to invent novel ways of applying the
  region-growing framework in \cite{lr99} (see
  \Cref{subsec:regiongrowing}). Many steps in our algorithms
  rely on such region growing arguments, including 
  \Cref{lem:singlescale,lem:coverfail,lem:ARV}.
  
\item In \cite{ARV} all vectors have unit norm, here however the
  $\ell_2^2$ vectors can have different norms. We use a known reduction
  that transforms the vectors for a large subset of vertices, so that
  they are in a sphere of fixed radius. See
  \Cref{subsec:singlescale}.

\item The existence of matching covers used in the ARV proof is
  unclear and has to be carefully established. This uses a certain
  ``spreading constraint'' that holds for $\ell_2^2$ metrics supported
  on small sets. See \Cref{lem:gaussian}. Also, a matching cover
  may not exist because a set of vertices is far away from other
  vertices in graph distance (distance according to the weights on
  edges). We call such sets {\em obstacle sets of type I}, and use
  region-growing arguments to remove these sets, see
  \Cref{subsec:obstacle}.

\item The crux of the ARV proof is to prove the existence of a special
  pair of vertices  that are close in graph metric (i.e., the metric
  given by the weights on the edges) and far apart in $\ell_2^2$
  metric). From the existence proof and global expansion
  $\Phi_{global}$, one can immediately establish the existence of
  $\Omega(n)$ such pairs, which is needed in the argument. The
  analogous idea does not work here since the proof is handicapped by
  being restricted to only use local expansion. However, we show that this step can only
  fail if there exists an {\em obstacle set of
  type II}. We design another region-growing type argument to handle
this; see~\Cref{lem:coverfail,lem:ARV,lem:obstacleII}.

\item The ARV argument uses Alon-Cheeger inequality: for $d$ regular
  graphs, the second eigenvalue of the Laplacian is $\Omega(1)$ iff
  the graph has expansion $\Omega(1)$. The analogous result for small
  set expansion, the so-called ``higher order Cheeger inequality,''
  has only recently been established, and only in one direction and in
  a weaker form \cite{lrtb12, lgt12}.  This weak form makes us lose
  extra $poly(\log r)$ factors in many theorems which are potentially
  improveable.
  For details see \Cref{lem:combtospectral}.
\end{enumerate}

\subsection{Finding Sparsest Cut using SSE flow}

Before we delve into the long proof of existence of SSE flows, we
quickly show how they are useful in approximating {\sc sparsest cut}.
As mentioned, there are two methods for this.

\subsubsection{Method 1: Using Lasserre Hierarchy Relaxation}
\label{subsubsec:lasserreSSE}

This will use a modification of an idea of Guruswami-Sinop which we
now recall.  Recall (see \Cref{sec:overview-sdp}) that the solutions
for $r'+2$ rounds of Lasserre Hierarchy relaxation satisfies the
following property:
$$
\frac{\sum_{u<v} C_{uv} \|\xmat_u-\xmat_v\|^2}{\sum_u \|\xmat_u\|^2} =
\frac{\tr(X^TXL(G))}{\|X\|_F^2} = \phi_{SDP},
$$
where the approximation ratio is bounded by $(1-\frac{\|X_S^\perp
  X\|_F^2}{\|X\|_F^2})^{-1}$ over all sets $S$ of size $r'$
by~\Cref{thm:gs-round-or-fail}.

\begin{theorem}[Theorem 3.2 in \cite{gs11-exp}] \label{thm:GS13-2} Given
  positive integer $r \ge 0$ and positive real $\eps>0$, the above
  approximation ratio is upperbounded by $\left(1 - \frac{1}{1-\eps}
    \frac{\sum_{i> r} \sigma_i(\xmat^T
      \xmat)}{\|X\|_F^2}\right)^{-1}$ for $r' = \frac{r}{\eps}+r+1$.
\end{theorem}
\begin{proof} (Sketch) Using the column based low-rank matrix
  reconstruction error bound from~\cite{gs11-svd}, it can be shown
  that there exists set $S$ of size $r' = r/\epsilon + r - 1$ such
  that the numerator $\|X_S^\perp X\|_F^2 \le (1-\epsilon)^{-1}
  \sum_{j\ge r+1} \sigma_j(X^TX)$, where $\sigma_j(X^TX)$ is the
  $j^{th}$ largest eigenvalue of $X^TX$.
\end{proof}
In order to bound the sum of eigenvalues, the analysis
in~\cite{gs11-exp} uses von Neumann's trace inequality, which we
present in a slightly more general form:
\begin{proposition} \label{thm:lb-on-sigmar} For any matrix $Y\succeq
  0$ and positive integer $r$:
  \[
  \sum_{i\ge r+1} \sigma_i(Y) = 
  \min_{Z\succeq 0} \frac{\tr(Y \cdot Z)}{\lambda_{r+1}(Z)}.
  \]
\end{proposition}
In the original analysis of~\cite{gs11-exp}, this claim is used with
$Y\gets \xmat^T \xmat$ and $Z\gets L(G)$, whereupon one obtains:
\[
\frac{\sum_{i> r} \sigma_i(\xmat^T \xmat)}{\|X\|_F^2} \le \frac{\tr(\xmat^T \xmat\cdot
  L(G))}{\lambda_{r+1}(G)\|X\|_F^2} \le \frac{\phi_{SDP}}{\lambda_{r+1}(G)}.
\]
Thus 
\begin{align}
\frac{\sum_{i\le r} \sigma_i(\xmat^T \xmat)}{\|X\|_F^2} \ge
1-\frac{\phi_{SDP}}{\lambda_{r+1}(G)}.  
\label{eq:4}
\end{align}
Consequently, the rounding analysis in~\cite{gs11-exp} requires a
bound on the $\lambda_{r+1}$ value of the graph.

Our idea is to use~\Cref{thm:lb-on-sigmar} by substituting the
Laplacian of the spectral SSE flow as $Z$ in the above calculation,
and then use the lowerbound on the $\lambda_r$ value of this flow
Laplacian. This uses the following easy lemma.
%
%
%
%
\begin{lemma}
  \label{thm:sigma-r-lb-by-flows} If $\xmat$ is described above, then for
  for any flow $ F$ that lies in the host graph  $G$:
  \[
  \frac{\sum_{i> r} \sigma_i(\xmat^T \xmat)}{\|X\|_F^2} \le \frac{\phi_{SDP}}{\lambda_{r+1}(F)}.
  \]
\end{lemma}
\begin{proof}
  Since $F$ is routable in $G$ and $\xmat\in\ell^2_2$, we have:
  \[ \tr(\xmat^T \xmat \cdot L(F)) \le \tr(\xmat^T \xmat\cdot L(G))
  \le \phi_{SDP} \|X\|_F^2.\] Choosing $Y\gets \xmat^T \xmat$ and $Z\gets
  L(F)$, we see that the Claim implies:
  \[
  \frac{\sum_{i> r} \sigma_i(\xmat^T \xmat)}{\|X\|_F^2} \le
  \frac{\tr(\xmat^T \xmat \cdot L(F))}{\lambda_{r+1}(F)\|X\|_F^2} \le
  \frac{\phi_{SDP}}{\lambda_{r+1}(F)}.\tag*{\qedhere}
  \] 
\end{proof}
Now \Cref{thm:SSE1} follows using \Cref{thm:sigma-r-lb-by-flows}~and~\cref{eq:4}.

\medskip \noindent{\bf Remark:} Note that we only need
$\lambda_{r+1}(F)$ to be more than $\phi_{SDP}$. Such flows could
potentially exist under more general conditions than our local
expansion condition.



%
\subsubsection{Method 2: Using Subspace Enumeration and Cut
  Improvement} 
We show that given a $(r,d,\lambda)$ spectral SSE flow, where $d\lambda$ is much
larger than the expansion $\Phi$ of sparsest cut, it is possible to
use eigenspace enumeration idea of~\cite{ABS} together with the ideas of~\cite{al08} to
get a good approximation to {\sc sparsest cut.}

\begin{lemma}[Eigenspace Enumeration, \cite{ABS}]
  \label{lem:eigenenumeration}
  There is a $2^{O(r)}n^{O(1)}$ time algorithm that, given a graph
  whose $r^{th}$ smallest eigenvalue (of Laplacian) is $\lambda_r \ge
  20\Phi/\epsilon$, outputs a set of subsets $X\subset \{0,1\}^V$ with
  the following guarantee: for every subset $S$ that has expansion
  $\Phi$, there is a vector $x\in X$ such that
  $$
  \frac{|x-\vec{1}_S|}{|\vec{1}_S|} \le \frac{8\Phi}{\lambda_r}.
  $$
\end{lemma}

The above eigenspace enumeration allows us to compute a ``guess'' that
has low symmetric difference with 
the optimum cut. Then we can use a simple version of
cut improvement algorithm of~\cite{al08} to improve it to a cut of low
expansion.
\begin{lemma}
  \label{lem:improvecut}
  There is a $2^{O(r)}n^{O(1)}$ time algorithm that given a graph
  $G=(V,E)$, and a $(r,d,\lambda)$ spectral SSE flow embeddable in
  $G$, enumerates $2^{O(r)}n^{O(1)}$ sets with the following
  guarantee.  For any set $S$ of size at most $n/2$ that has expansion
  $\Phi(S) \ll d\lambda \epsilon \delta$ (for $\epsilon+ \delta < 1$),
  there is a set $T$ in the output such that $\frac{|T\Delta S|}{|S|}
  \le \delta$ and $\Phi(T)\le (1+\epsilon)\Phi(S)$ ($\Delta$ denotes
  symmetric difference).
\end{lemma}

\begin{proof}
  The capacity of flow that crosses $S$ in the spectral SSE flow can
  only be smaller than $\Phi(S)\cdot |S|$ because the flow is
  embeddable in $G$. Hence when we apply
  \Cref{lem:eigenenumeration} on the flow, we know there is a
  vector $\vec{1}_T$ in $X$ that is $\epsilon\delta/2$ close to the
  indicator vector of $S$.

  Using this vector, suppose we know the expansion $\Phi(S)$ (later we
  shall see we only need to know this value up to multiplicative
  factor, so the algorithm will enumerate all possible
  values). Construct a single commodity flow instance where we add a
  source $s$ and sink $t$ to the graph. For each vertex $i\in T$,
  there is an edge from $i$ to sink $t$ with capacity
  $4\Phi(S)/\delta$. For each vertex $i\not\in T$, there is an edge
  from source $s$ to $i$ with capacity $4\Phi(S)/\delta$.

  Now we find the min-cut that separates source $s$ and sink
  $t$. Since $T$ is close to $S$, we know the capacity of this cut is
  at most $(1+\epsilon/2)\Phi(S)|S|$ because $\{s\}\cup S$ achieves
  this capacity. Let the vertices that are on the same side with sink
  be $Q$, then we know $|Q\Delta T| \le
  \frac{(1+\epsilon/2)\Phi(S)|S|}{4\Phi(S)/\delta} \le
  |S|\delta/2$. Therefore $\frac{|Q\Delta S|}{|S|} \le \frac{|Q\Delta
    T| + |T\Delta S|}{|S|} \le \delta$.

  On the other hand, the expansion of $Q$ is at most
  $$\frac{(1+\epsilon/2)\Phi(S)|S| - |Q\Delta T|\cdot
    4\Phi(S)/\delta}{|S| - |S\Delta T| - |Q\Delta T|} 
  = \frac{(1+\epsilon/2)\Phi(S) -
    4x\Phi(S)/\delta}{(1-\epsilon\delta/2) - x} 
  \le (1+\epsilon)\Phi(S).$$
  (where in the second step, we substituted $x \triangleq \frac{|Q\Delta T|}{|S|}$).
\end{proof}

\begin{corollary}
  Given graph $G = (V,E)$ and a $(r,d,\lambda)$ spectral SSE flow
  embeddable in $G$. There is a $2^{O(r)}n^{O(1)}$ time algorithm that:
  \begin{itemize}
  \item Finds a set $S$ with $\phi(S) \le
    (1+O(\epsilon))\phi_{sparsest}$ if $d\lambda \gg
    \phi_{sparsest}/\epsilon^2$;
  \item Finds a set $S$ with $\Phi(S) \le
    (1+O(\epsilon))\Phi_{global}$ if $d\lambda \gg
    \Phi_{global}/\epsilon$;
  \item Finds a set $S$ of size at least $cn/2$ such that $\Phi(S) \le
    (1+O(\epsilon))\Phi_{c\mbox{-balanced}}$ if $d\lambda
    \gg\Phi_{c\mbox{-balanced}}/c\epsilon$.
  \end{itemize}
\end{corollary}
\begin{proof} (sketch) For sparsest cut, choose $\delta = \epsilon$ in
  \Cref{lem:improvecut}. For edge expansion, choose $\delta =
  1/2$. For $c$-balanced separator, choose $\delta = c/2$.
\end{proof}
%

%
\section{Conclusions}
The fact that it is possible to compute $(1+\epsilon)$-approximation
for {\sc sparsest cut} on an interesting family of graphs seems very
surprising to us.  Further study of Guruswami-Sinop rounding also
seems promising: our analysis is still not using the full power of
their theorem.

Our work naturally leads us to the following imprecise conjecture,
which if true would yield immediate progress.

\medskip

\noindent {\bf Conjecture:} (Imprecise) {\em In ``interesting''
  families of graphs ---ie those where existing algorithms for {\sc
    sparsest cut} fail--- $\Phi_{local}/\Phi_{global}$ is large, say
  $\gg \sqrt{\log n}$.}

As support for this conjecture we observe that if our algorithm does
not beat $\sqrt{\log n}$-approximation on some graph, then there is a
constant $r$ and a set of size $n/r$ whose expansion is at least
$\sqrt{\log n}$ times the optimum.

Furthermore, it is conceivable that SSE flows exist in graphs even
when the local expansion condition is not met. For our analysis of the
rounding algorithm from~\cite{gs11-exp} we only need the existence of
an SSE flow of degree say $> 1.1 \phi_{sparsest}$
(see~\Cref{subsubsec:lasserreSSE}). Conceivably such flows exist in a
wider family of graphs, and this could be another avenue for progress.
\section*{Acknowledgements}
We gratefully acknowledge helpful discussions with Venkat Guruswami,
Jon Kelner, Ravi Krishnaswamy, Assaf Naor and David Steurer. This work
was funded by grants from the NSF and the Simons Foundation.



\bibliographystyle{alpha}
\bibliography{references}

\newcommand{\etalchar}[1]{$^{#1}$}
\begin{thebibliography}{BHK{\etalchar{+}}12}

\bibitem[ABS10]{ABS}
Sanjeev Arora, Boaz Barak, and David Steurer.
\newblock Subexponential algorithms for {Unique} {Games} and related problems.
\newblock In {\em FOCS}, pages 563--572, 2010.

\bibitem[AHK04]{AHK1}
Sanjeev Arora, Elad Hazan, and Satyen Kale.
\newblock $o(\sqrt {\log n})$ approximation to sparsest cut in $\tilde{O}(n^2)$
  time.
\newblock In {\em FOCS}, pages 238--247, 2004.

\bibitem[AK07]{AK}
Sanjeev Arora and Satyen Kale.
\newblock A combinatorial, primal-dual approach to semidefinite programs.
\newblock In {\em STOC}, pages 227--236, 2007.

\bibitem[AL08]{al08}
Reid Andersen and Kevin~J. Lang.
\newblock An algorithm for improving graph partitions.
\newblock In {\em SODA}, pages 651--660, 2008.

\bibitem[ALN08]{aln}
Sanjeev Arora, James Lee, and Assaf Naor.
\newblock Euclidean distortion and the sparsest cut.
\newblock {\em J. American Mathematical Society}, 21(1):1--21, 2008.

\bibitem[AM85]{am85}
Noga Alon and V.~D. Milman.
\newblock lambda$_{\mbox{1}}$, isoperimetric inequalities for graphs, and
  superconcentrators.
\newblock {\em J. Comb. Theory, Ser. B}, 38(1):73--88, 1985.

\bibitem[ARV09]{ARV}
Sanjeev Arora, Satish Rao, and Umesh~V. Vazirani.
\newblock Expander flows, geometric embeddings and graph partitioning.
\newblock {\em J. ACM}, 56(2), 2009.

\bibitem[BFK{\etalchar{+}}11]{bansalsse11}
Nikhil Bansal, Uriel Feige, Robert Krauthgamer, Konstantin Makarychev,
  Viswanath Nagarajan, Joseph Naor, and Roy Schwartz.
\newblock Min-max graph partitioning and small set expansion.
\newblock In {\em FOCS}, pages 17--26, 2011.

\bibitem[BHK{\etalchar{+}}12]{bhksz12}
Boaz Barak, Aram Harrow, Jonathan Kelner, David Steurer, and Yuan Zhou.
\newblock Hypercontractivity, {Sum-of-Squares} proofs, and their applications.
\newblock In {\em STOC}, pages 307--326, 2012.

\bibitem[BRS11]{brs11}
Boaz Barak, Prasad Raghavendra, and David Steurer.
\newblock Rounding semidefinite programming hierarchies via global correlation.
\newblock In {\em FOCS}, pages 472--481, 2011.

\bibitem[CMM06]{cmm06}
Eden Chlamtac, Konstantin Makarychev, and Yury Makarychev.
\newblock How to play unique games using embeddings.
\newblock In {\em FOCS}, pages 687--696, 2006.

\bibitem[DKSV06]{dksv06}
Nikhil~R. Devanur, Subhash Khot, Rishi Saket, and Nisheeth~K. Vishnoi.
\newblock Integrality gaps for sparsest cut and minimum linear arrangement
  problems.
\newblock In {\em STOC}, pages 537--546, 2006.

\bibitem[DR10]{dr10}
Amit Deshpande and Luis Rademacher.
\newblock Efficient volume sampling for row/column subset selection.
\newblock In {\em FOCS}, pages 329--338, 2010.

\bibitem[FT03]{ft03}
Jittat Fakcharoenphol and Kunal Talwar.
\newblock An improved decomposition theorem for graphs excluding a fixed minor.
\newblock In {\em RANDOM-APPROX}, pages 36--46, 2003.

\bibitem[GS11]{gs11-qip}
Venkatesan Guruswami and Ali~Kemal Sinop.
\newblock {L}asserre hierarchy, higher eigenvalues, and approximation schemes
  for graph partitioning and quadratic integer programming with {PSD}
  objectives.
\newblock In {\em FOCS}, pages 482--491, 2011.

\bibitem[GS12a]{gs12-fast}
Venkatesan Guruswami and Ali~Kemal Sinop.
\newblock Faster {SDP} hierarchy solvers for local rounding algorithms.
\newblock In {\em FOCS}, pages 197--206, 2012.

\bibitem[GS12b]{gs11-svd}
Venkatesan Guruswami and Ali~Kemal Sinop.
\newblock Optimal column-based low-rank matrix reconstruction.
\newblock In {\em SODA}, pages 1207--1214, 2012.

\bibitem[GS13]{gs11-exp}
Venkatesan Guruswami and Ali~Kemal Sinop.
\newblock Approximating non-uniform sparsest cut via generalized spectra.
\newblock In {\em SODA}, 2013.

\bibitem[GT13]{gt13}
Shayan~Oveis Gharan and Luca Trevisan.
\newblock {Improved ARV Rounding in Small-set Expanders and Graphs of Bounded
  Threshold Rank}.
\newblock {\em ArXiv e-prints}, April 2013.

\bibitem[KLL{\etalchar{+}}13]{KLLGT}
Tsz~Chiu Kwok, Lap~Chi Lau, Yin~Tat Lee, Shayan Oveis~Gharan, and Luca
  Trevisan.
\newblock Analysis of spectral partitioning through higher order spectral gap.
\newblock In {\em STOC}, 2013.

\bibitem[Las02]{Las02}
Jean~B. Lasserre.
\newblock An explicit equivalent positive semidefinite program for nonlinear
  0-1 programs.
\newblock {\em SIAM J. Optimization}, 12(3):756--769, 2002.

\bibitem[LGT12]{lgt12}
James~R. Lee, Shayan~Oveis Gharan, and Luca Trevisan.
\newblock Multi-way spectral partitioning and higher-order cheeger
  inequalities.
\newblock In {\em STOC}, pages 1117--1130, 2012.

\bibitem[LR99]{lr99}
Frank~Thomson Leighton and Satish Rao.
\newblock Multicommodity max-flow min-cut theorems and their use in designing
  approximation algorithms.
\newblock {\em J. ACM}, 46(6):787--832, 1999.

\bibitem[LRTV12]{lrtb12}
Anand Louis, Prasad Raghavendra, Prasad Tetali, and Santosh Vempala.
\newblock Many sparse cuts via higher eigenvalues.
\newblock In {\em STOC}, pages 1131--1140, 2012.

\bibitem[LS10]{ls10}
James~R. Lee and Anastasios Sidiropoulos.
\newblock Genus and the geometry of the cut graph.
\newblock In {\em SODA}, pages 193--201, 2010.

\bibitem[LS11]{ls11}
James~R. Lee and Anastasios Sidiropoulos.
\newblock Near-optimal distortion bounds for embedding doubling spaces into l1.
\newblock In {\em STOC}, pages 765--772, 2011.

\bibitem[MMV12]{mmv}
Konstantin Makarychev, Yury Makarychev, and Aravindan Vijayaraghavan.
\newblock Approximation algorithms for semi-random partitioning problems.
\newblock In {\em STOC}, pages 367--384, 2012.

\bibitem[MN04]{mn04}
Manor Mendel and Assaf Naor.
\newblock Euclidean quotients of finite metric spaces.
\newblock {\em Advances in Mathematics}, 189(2):451--494, 2004.

\bibitem[OZ13]{oz13}
Ryan O'Donnell and Yuan Zhou.
\newblock Approximability and proof complexity.
\newblock In {\em SODA}, 2013.

\bibitem[Par03]{parrilo03}
Pablo~A. Parrilo.
\newblock Semidefinite programming relaxations for semialgebraic problems.
\newblock {\em Math. Program.}, 96(2):293--320, 2003.

\bibitem[PP93]{pp93}
James~K. Park and Cynthia~A. Phillips.
\newblock Finding minimum-quotient cuts in planar graphs.
\newblock In {\em STOC}, pages 766--775, 1993.

\bibitem[RS10]{rs10-sse}
Prasad Raghavendra and David Steurer.
\newblock Graph expansion and the unique games conjecture.
\newblock In {\em STOC}, pages 755--764, 2010.

\bibitem[She09]{Sherman}
Jonah Sherman.
\newblock Breaking the multicommodity flow barrier for $o(\sqrt{\log
  n})$-approximations to sparsest cut.
\newblock In {\em FOCS}, pages 363--372, 2009.

\end{thebibliography}

\appendix


 




\section{Overview of Lasserre Hierarchy Relaxation for Sparsest Cut
  and Rounding}
\label{sec:overview-sdp}
In this section, we will give a brief description of Lasserre
Hierarchy relaxation for Uniform Sparsest Cut problem, the rounding
algorithm of~\cite{gs11-exp} and its analysis. 

We present the formal definitions of Lasserre Hierarchy
relaxations~\cite{Las02}, tailored to the setting of the problems we
are interested in, where the goal is to assign to each node in $V$ a
label from $\{0,1\}$.
\def\dim{\Upsilon}
\begin{definition}[Lasserre vector set]
\label{def:las-sdp}
Given a set of variables $V$ and a positive integer $r$, a collection
of vectors $\xvec$ is said to satisfy $r$-rounds of Lasserre
Hierarchy, denoted by $\xvec \in
\lasserreii{r}{V}$, 
if it satisfies the following conditions:
\begin{enumerate}
\item For each set $S\in\binom{V}{\le r+1}$, there exists a function
  $\xvec_S:\{0,1\}^{S}\to \R^\dim$ that associates a vector of some
  finite dimension $\dim$ with each possible labeling of $S$.  We use
  $\xvec_S(f)$ to denote the vector associated with the labeling $f\in
  \{0,1\}^S$.
  For singletons $u\in V$, we will use $\xvec_u$ and $\xvec_u(1)$
  interchangeably.
  For $f \in \{0,1\}^S$ and $v\in S$, we use $f(v)$ as the label $v$
  receives from $f$.  Also given sets $S$ with labeling $f\in
  \{0,1\}^S$ and $T$ with labeling $g\in\{0,1\}^T$ such that $f$ and
  $g$ agree on $S\cap T$, we use $f\circ g$ to denote the labeling of
  $S\cup T$ consistent with $f$ and $g$: If $u\in S$, $(f\circ g)(u) =
  f(u)$ and vice versa.
\item $\|\xvec_{\es}\|^2 = 1$.
\item $\langle \xvec_S(f), \xvec_T(g)\rangle = 0$ if there exists
  $u\in S\cap T$ such that $f(u)\neq g(u)$.
\item \label{def:las-sdp:consistent} $\langle \xvec_S(f),
  \xvec_T(g)\rangle = \langle \xvec_A(f'), \xvec_B(g')\rangle$ if
  $S\cup T=A\cup B$ and $f\circ g = f'\circ g'$.
\item For any $u\in V$, $\sum_{j\in \{0,1\}} \|\xvec_u(j)\|^2 =
  \|\xvec_\es\|^2$.
\item (implied by above constraints) For any $S\in\binom{V}{\le r+1}$,
  $u\in S$ and $f\in \{0,1\}^{S\setminus\{u\}}$, $\sum_{g\in
    \{0,1\}^u} \xvec_{S}(f\circ g) = \xvec_{S\setminus\{u\}} (f)$.
\end{enumerate}
\end{definition}
One can view $\|\xvec_{S}(f)\|^2$ as the ``probability'' of $f$, in
which case the corresponding ``conditional'' probabilities are given
by $\frac{ \langle \xvec_S(f), \xvec_u \rangle }{\|\xvec_S(f)\|^2}$.
Our relaxation is the following:
\begin{equation}
\label{eq:sc-sdp-alg}
\begin{array}{rll}
  \min_{\mu,\xvec} & \frac{1}{n \mu (1-\mu)} \sum_{u<v} C_{u,v} \| \xvec_u - \xvec_v \|^2 \\
  \mathrm{st} & \frac{1}{n} \sum_u \xvec_u = \mu \xvec_{\es}, \\
  & \|\xvec_{\es}\|^2 = 1,\quad \xvec \in \lasserreii{r'+2}{V},\quad
  \mu \in \{1/n,2/n,\ldots,1/2\}.
\end{array}
\end{equation} 
%
%
Note that we can easily eliminate the variable $\mu$
from~\cref{eq:sc-sdp-alg} by enumerating over all $\frac{n}{2}$
values.
For the special case of uniform sparsest cut, the rounding algorithm
from~\cite{gs11-exp} can be summarized as follows. Given a feasible
solution of~\cref{eq:sc-sdp-alg}:
\begin{enumerate}
\item Let $ \xmat_u \triangleq \xvec_u - \frac{1}{n} \sum_v \xvec_v =
  \xvec_u - \mu \xvec_{\es}$.
  Observe
  \begin{inparaenum} [(i)]   
  \item $\|\xmat_u - \xmat_v\|^2 = \|\xvec_u-\xvec_v\|^2$;
  \item $\sum_u \xmat_u = 0$;
  \item $\sum_u \|\xmat_u\|^2 = \frac{1}{n} \sum_{u<v} \|\xvec_u -
    \xvec_v\|^2 = \mu (1-\mu)$ ;
  \item $1-\mu\ge \|\xmat_u\|^2 \ge \mu^2$.
  \end{inparaenum}
\item Choose a set $S$ of $r'$ nodes using column
  selection~\cite{gs11-svd,dr10} from $[\xmat_u]_u$.
\item Sample $f: S \to \{0,1\}$ with probability proportional to
  $\|\xvec_S(f)\|^2$.
\item Perform threshold rounding using the ``conditional
  probabilities'' assigned for each node $u\in V$ which is
  proportional to $\langle \xvec_{S}(f), \xvec_u \rangle$.
\end{enumerate}

\section{Constructing SSE Flows}

\subsection{Definition and LP Formulation of SSE Flows}
\label{subsec:LPformulation}


Recall for any graph $G =
(V,E)$ 
, a multicommodity flow in $G$ assigns demand $\delta_{i,j}$ for pairs
of vertices $i,j$, and simultaneously route $\delta_{i,j}$ units of
flow from $i$ to $j$ for all pairs while satisfying capacity
constraints. In particular, if we use $f_p$ to denote the amount of
flow routed along path $p$, $\mathcal{P}_{i,j}$ to denote all paths
from $i$ to $j$, then a multicommodity flow should satisfy the
following constraints:
\begin{align}
  \forall i,j \in V\ \quad \sum_{p\in \mathcal{P}_{i,j}} f_p & = \delta_{i,j} \\
  \forall e\in E \quad \sum_{e\in p} f_p & \le c_e
\end{align}
We shall only consider symmetric flows (i.e. $\delta_{i,j} =
\delta_{j,i}$). For a multicommodity flow, we call $\delta_i =
\sum_{j\in V} \delta_{i,j}$ the degree of vertex $i$. Now we can
define SSE flows:

\begin{definition}[SSE flow]
\label{def:sseflow}
A $(r, d, \beta)$-SSE flow is a multicommodity flow whose vertices
have degree at most $d$, and for any set $S$ of size at most $n/r$, we
have
$$\sum_{i\in S, j\not\in S} \delta_{i,j} \ge \beta d |S|.$$
\end{definition}

For a flow, we will define the expansion of a set $S$ to be
$\frac{\sum_{i\in S,j\not \in S} \delta_{i,j}}{d|S|}$, so the
requirement of SSE flow is just the expansion of all small sets should
be at least $\beta$.

SSE flow is a complicated object, we will also use two weaker versions
of SSE flow. The first one is especially useful for the LP
formulation:

\begin{definition}[Weak SSE flow]
\label{def:weaksseflow}
A $(r, d, \beta)$ weak SSE flow is a multicommodity flow whose
vertices have degree at most $d$, and for any set $S$ of size between
$n/3r$ and $n/r$, we have
$$\sum_{i\in S, j\not\in S} \delta_{i,j} \ge \beta d |S|.$$
\end{definition}
Notice that the idea of restricting set $S$ to have roughly size $n/r$
is also used in \cite{ARV} (where in the LP formulation the sets have
size $n/6$ to $n/2$).  We will use the LP formulation for weak SSE
flows:
\begin{align}
  \forall i\in V \quad & \sum_j \sum_{p\in \mathcal{P}_{i,j}} f_p \le d \label{eqn:degree}\\
  \forall e\in E \quad & \sum_{e\in p} f_p \le c_e \label{eqn:capacity}\\
  \forall S\subset V, n/3r\le |S| \le n/r \quad & \sum_{i\in S,
    j\not\in S} \sum_{p\in \mathcal{P}_{i,j}} f_p \ge \beta d
  |S| \label{eqn:expansion}
\end{align}
Later we show this LP is feasible for some values of $\beta$ and $d$
by showing that the dual LP is not feasible. The dual LP is given by
\begin{align}
  \sum_e c_e w_e + d \sum_{i\in V} s_i  & < \beta d \sum_{S} z_S |S| \label{eqn:dualvalue}\\
  \forall i,j, p\in \mathcal{P}_{i,j} \qquad
  \sum_{e\in p} w_e + s_i + s_j & \ge \sum_{S:i\in S, j\not \in S} z_S \label{eqn:path}\\
  z_S, w_e, s_i & \ge 0 \label{eqn:dualnonnegative}
\end{align}
Here $s_i$ ($i\in V$), $w_e$ ($e\in V$) and $z_S$ ($S\subset V,
n/3r\le |S|\le n/r$) are the dual variables corresponding to
\cref{eqn:expansion,,eqn:capacity,,eqn:degree} respectively. Since the
dual LP is homogeneous, all variables can be scaled simultaneously
without effecting the validity of a solution. Throughout this section
we shall always assume the following normalization:
\begin{equation}
\sum_S z_S|S| = n \label{eqn:normalization}
\end{equation}
We would like SSE flows to serve as approximate certificate of small
set expansion. However, small set expansion is in general hard to
certify, even if the expansion is close to $1$. Hence we give a weaker
(but still useful) form, spectral SSE flow, which can be easily
certified, and is closely related to combinatorial SSE flows by high
order Cheeger's inequality\cite{lrtb12, lgt12}.
\begin{definition}
  A $(r, d,\lambda)$ spectral SSE flow is a multicommodity flow whose
  vertices have degree between $d/2$ and $d$, and the $r^{th}$
  smallest eigenvalues of the Laplacian of the graph is at least
  $d\lambda$.
\end{definition}
\subsection{Existence of SSE Flows}
The main result of this Section is the existence of weak SSE flows.
\begin{lemma}
\label{lem:weaksseflowexists}
For any graph $G = (V,E)$ and any $d>0$ either there is a set of size
at most $100n/r$ that has expansion smaller $d\sqrt{log n}/\sqrt{\log
  r}$, or there exists a $(r, d, \Omega(log^{-2} r))$ weak SSE flow
embeddable in $G$.
\end{lemma}
\subsubsection{$\ell_2^2$ Mapping of $z_S$}
\label{subsec:mapping}
In order to show that weak SSE flow exist, we argue that the dual LP
does not have any valid solution. In fact, we show something even
stronger: given any dual solution that satisfies 
\cref{eqn:dualvalue}, \cref{eqn:dualnonnegative} and 
\cref{eqn:normalization},
there is a polynomial time algorithm that
either finds a nonexpanding set of size at most $100n/r$, or finds a
path where \cref{eqn:path}
is violated.

The dual solution has exponentially many variables. We shall use a
compressed description that is good enough for the algorithm: all the
variables $z_S$ are mapped to $n$ vectors $Z_1,...,Z_n$, such that for
all $i,j$, $\norm{Z_i-Z_j}^2 = \sum_{i\in S, j\not\in S} z_S$. This is
possible because $Z_i$'s can have one coordinate for each set $S$, if
$i$ is in $S$ then the coordinate is $\sqrt{z_S}$, otherwise the
coordinate is just 0.

The upper-bound on the size of $S$ implies for any $i$, there are at
most $Cn/r$ points within $\ell_2^2$ distance $\norm{Z_i}^2/C$ for all
$C > 1$. This is because for any set $S$ of $Cn/r$ points, if we pick
a random point $j$ in $S$, with probability $1$ the expected distance
between $i$ and $j$ is at least $\E_{j\in S}[\norm{Z_i-Z_j}^2] \ge
\sum_{t\ni i} z_S\Pr[j\not \in t] \ge \sum_{t\ni i} z_S/C =
\norm{Z_i}^2/C$. To avoid giving out exponentially many variables, the
dual solution will only contain $Z_i$'s that satisfy this property.

We shall rewrite the constraints for $Z_i$'s

\begin{align}
  \sum_e c_e w_e + d \sum_{i\in V} s_i  & < \beta d n \label{eqn:newdualvalue}\\
  z_S, w_e, s_i &  \ge 0 \label{eqn:newdualnonneg} \\
  \sum_{i\in V} \norm{Z_i}^2  & = n \label{eqn:newnorm} \\
  \forall i\in V, C>1 \qquad |\{j: \norm{Z_j-Z_i}^2 \le 4\norm{Z_i}^2/5\}| &\le 5n/4r \label{eqn:spreading} \\
  \forall i\in V \qquad \norm{Z_i}^2 & \le 3r \label{eqn:normZ} \\
  \forall i,j\in V, p\in \mathcal{P}_{i,j} \qquad \sum_{e\in p} w_e +
  s_i+s_j & \ge \norm{Z_i-Z_j}^2 \label{eqn:newpath}
\end{align}



The last constraint is called the {\em spreading constraint}. A
candidate dual solution is just a set of variables $s_i(i\in V),
w_e(e\in E), Z_i(i \in S)$ that satisfies these constraints.

Consider $w_e$'s as edge distances on the graph, and let $d_{i,j}$ be
the shortest path distance between $i$ and $j$ with weights
$w_e$. From now on we refer to this weighted distance as the {\em
  graph distance}. Intuitively, given a candidate dual solution, the
algorithm either finds a nonexpanding set or finds two vertices who
have small $s_i$'s, small distance in graph distance and large
distance in $\ell_2^2$ distance.

\subsubsection{Proof Idea}
\label{subsec:proofidea}

Given the dual solution, we try to apply arguments similar to
\cite{ARV} in order to find a pair of vertices that are close in graph
distance, but far in $l_2^2$ metric. When this pair of vertices also
have small $s_i$'s, it violates \Cref{eqn:newpath} hence contradicting
the feasibility of dual solution.

In~\cite{ARV} this proof goes by projecting all points along a random
direction and arguing that there must be many pairs of points that are
close in graph distance but far in projection distance: they call it a
{\em matching cover}. However, constructing a matching cover in
our setting case is highly nontrivial, because the proof is only
allowed to use local expansion. We adapt the region-growing argument
in \cite{lr99} in novel ways to solve this problem, see
\Cref{subsec:regiongrowing}.

The first difficulty in the argument is that each vertex might have
very different $\norm{Z_i}$ (that is, they are in very different
measure of sets in the dual solution). But this is easily fixed by
embedding the points into a single scale using ideas from \cite{aln}
(see \Cref{lem:singlescale} in \Cref{subsec:singlescale}).

It turns out that in order for a matching cover to not exist, one of
the following two types of {\em obstacles} must exist, detailed
discussion appears in \Cref{subsec:obstacle}.

The first type of obstacle set is a set whose $D_0$-neighborhood in
graph distance ($D_0$ is a parameter that will be chosen later)
contains only $O(n/r)$ points. Intuitively, this is an obstacle since
it would be hard to match these vertices to other vertices within
graph distance $D_0$ because they simply don't have enough
neighbours. We show that the total volume of such sets cannot be too
large using the region-growing framework, see \Cref{cor:expansion}.

The second type of obstacle set is a set of at most $10n/r$ vertices
with large $s_i$ whose $D_0$ neighbourhood in graph distance contains
only $O(n/r)$ points with small $s_i$. Intuitively such sets are bad
because we want to construct matching covers only on vertices with
small $s_i$ (in order to get the final contradiction
with~\Cref{eqn:newpath}). Such sets would mean it is possible for a set
$S$ with small $s_i$ to be only close to vertices with large $s_i$'s,
and it would be impossible to match all the vertices in $S$ with
vertices with small $s_i$'s. We again use region-growing arguments to
remove such sets. The number of vertices removed cannot be large,
because otherwise the sum of $s_i$'s will be too large and violates
\cref{eqn:dualvalue} (see \Cref{lem:obstacleII}).

Without these obstacle sets, it becomes possible to construct a
matching cover (see \Cref{lem:coverfail} in
\Cref{subsec:matchingcover}). This matching cover allows us to adapt arguments in \cite{ARV} (see
\Cref{lem:bigcover,lem:ARV}), and conclude that either there is a short path (in graph distance) that crosses many cuts, or there is a nonexpanding set. The first case contradicts with the validity of the dual solution. In the second case we get a nonexpanding set, which again implies the existence of obstacle sets of type I or II.

\subsubsection{Region-growing Argument}
\label{subsec:regiongrowing}
As mentioned earlier, a key component of our proof is the
region-growing argument from \cite{lr99}. This argument applies to an
undirected graph whose edges have arbitrary nonengative
capacities. The goal (in \cite{lr99}) is to give a partition into
blocks that have low {\em diameter} (distance being measured using
edge weights) and on average have {\em few} edges crossing between the
blocks. The tradeoff between these two quantities is controlled by the
expansion of the underlying unweighted graph. Here we view this
argument as giving an efficient {\em partition oracle}, which
maintains a set of vertices $V_i$ at step $i$ ($V_0 = V$). At step $i$
the oracle takes a set $S_i\subset V_{i-1}$ of size at least $n/F(r)$
where $F$ is a fixed polynomial, and then outputs $S_i'$, $S_i\subset
S_i'\subset V_{i-1}$, and updates $V_i = V_{i-1}\backslash
S_i'$. There is a ``center'' $j \in S_i$ such that every other $j' \in
S_i$ has distance at most $D_0$ to $j$ (we specify $D_0$ later in
\Cref{lem:regiongrow}).

At step $t$, we say the partition maintained by the oracle is the
collection of disjoint sets $S_1', S_2', ..., S_i', V_i$. The capacity
of edges in the partition is always at most $n\alpha/20\Delta\log
30r$.
\begin{lemma}[\cite{lr99}]
\label{lem:regiongrow}
Let $G=(V,E)$ be a graph with edge capacities $c_e$ and edge lengths
$w_e$. Let $W$ be the total weighted edge length: $W = \sum_{e\in E}
c_ew_e$. Then for any polynomial $F(r)$, and any $D_0 = C\Delta\log
30r \cdot \log rW/n\alpha$ (where $C$ is a constant depending on $F$),
there is an efficient partition oracle whose partitions always have
capacity at most $n\alpha/20\Delta\log 30r$.
\end{lemma}
\begin{proof}
  The proof is similar to Lemma 3 in \cite{lr99}. However there the
  region-growing procedure starts from a single vertex (and the loss
  is $\log n$ because $n$ is roughly the ratio between the volume of
  the graph and the volume of a single vertex). Here instead we start
  region-growing from the sets given to the oracle. Because the sets
  all have large volume (more than $n/poly(r)$), we lose only a $\log
  r$ factor.
\end{proof}
%
%
%
\subsubsection{Reducing to Single Scale}
\label{subsec:singlescale}
The region growing argument applies to a particular scale
$\Delta$. However, not all vertices have $\ell_2^2$ norm close to that
scale. In this part we show how to reduce the problem to a single
scale $\Delta$.
\begin{lemma}\label{lem:singlescale}
  Given a dual solution with $\beta \le C_\beta(log r)^{-3/2}$, there
  is an algorithm that finds a $\Delta$, and calls a partition oracle
  with scale $\Delta$ and size $F(r)$. After the algorithm, the number
  of remaining vertices in the oracle is at least
  $\frac{n}{5\Delta\log 30r}$, and all but $n/F(r)$ of the remaining
  vertices satisfy one of the two properties:
  \begin{enumerate}
  \item $\norm{Z_i}^2 \ge \Delta/2$.
  \item $s_i \ge D_2/10 = \Omega(\Delta/\log r)$.
  \end{enumerate}

  The value $D_2$ comes from \Cref{lem:ARV} and will be
  $\Omega(\Delta/\log r)$.
\end{lemma}
First we use an averaging argument to find $\Delta$.
\begin{lemma}\label{lem:bucketing}
  There exists some threshold $0.1 < \Delta < 3r$ such that the number
  of vertices with $\norm{Z_i}^2\ge \Delta$ is at least
  $\frac{n}{4\Delta\log 30r}$.
\end{lemma}
\begin{proof}
  We shall bucket the vertices according to $\norm{Z_i}^2$. There will
  be $b = \lceil \log 30r\rceil$ buckets, the $u$-th ($u\in
  \{1,...,b\}$) bucket $B_u$ contains vertices with $\norm{Z_i}^2$ in
  range $[0.1*2^{u-1}, 0.1*2^u)$. There will be one extra bucket $B_0$
  which contains vertices with $\norm{Z_i}^2$ in range
  $[0,0.1)$. By~\Cref{eqn:normZ}, $\norm{Z_i}^2 \le 3 r$ so these
  buckets cover all vertices.

  We know $\sum \norm{Z_i}^2 = n$, let $l_u = \sum_{i\in B_u}
  \norm{Z_i}^2$, then $\sum_{u=0}^{b} l_u = n$.  We also know $l_0 \le
  0.1n$, so there must be a bucket $B_u$ with $l_u \ge 0.9n/b$. Choose
  $\Delta = 0.1*2^{u-1}$, we know the number of vertices with
  $\norm{Z_i}^2\ge \Delta$ is at least the size of $B_u$, which is at
  least $l_u / 2\Delta \ge n/4\Delta b$.
\end{proof}
Now we are ready to prove \Cref{lem:singlescale}.
\begin{proof}[Proof of~\Cref{lem:singlescale}.]
  Take the value $\Delta$ from~\Cref{lem:bucketing}. Let $Q$ be
  the set of vertices whose $\norm{Z_i}^2$ is at most $\Delta/2$ and
  $s_i$ is at most $D_2/10$. Let $B$ be the set of vertices whose
  $\norm{Z_i}^2$ is at least $\Delta$.

  If the size of $Q$ is at most $n/F(r)$, then the Lemma is
  true. Otherwise, use the partition oracle to separate the set
  $Q$. From the oracle we get a $Q'$ which contains everything in
  $Q$. Consider any vertex $i\in Q'\cap B$, by definition of oracle we
  know there is a vertex $j\in Q$ such that $D_{i,j} \le D_0$. On the
  other hand, $\norm{Z_i-Z_j}^2 \ge \norm{Z_i}^2-\norm{Z_j}^2 \ge
  \Delta/2$. By \cref{eqn:newpath} we know $s_i >
  \Omega(\Delta)$, since $\sum_{i\in V} s_i \le \beta n$, the size of
  $Q'\cap B$ is at most $\beta n/\Delta < |B|/50$. The size of current
  set of the oracle is at least $n/5\Delta \log 30r$.
\end{proof}
We shall also use the following Lemma from \cite{aln,mn04} to project
everything to a ball of squared radius $\Delta$. 
\begin{lemma}[\cite{mn04}] 
  There exists a mapping $T:\ell_2\to \ell_2$ such that $\norm{T(z)}
  \le \sqrt{\Delta}$ for all $z\in \ell_2$ and for all $z,z'\in
  \ell_2$
$$
\frac{1}{2} \le \frac{\norm{T(z)-T(z')}}{\min\{\sqrt{\Delta},
  \norm{z-z'}\}} \le 1.
$$
\end{lemma}
As a corollary, we now prove the following.
\begin{corollary}
  \label{cor:singlescalemapping}
  There is a mapping that maps $Z_i$ to $X_i$, such that $\norm{X_i} =
  \sqrt{\Delta/2}$ for all $i\in V$, and for all $i,j\in V$
  $$
  \frac{1}{8} \le \frac{\norm{X_i-X_j}^2}{\min\{\Delta,
    \norm{Z_i-Z_j}^2\}} \le 1.
  $$
\end{corollary}
\begin{proof}
  Just let $X_i =\frac{1}{\sqrt{2}} T(Z_i)\oplus (\sqrt{\Delta -
    \norm{T(Z_i)}^2})$ where $\oplus$ denotes concatenation of
  vectors. It is easy to verify the claim.
\end{proof}
After mapping all the $Z_i$'s to $X_i$'s, for a vertex $i$ with
$\norm{Z_i}^2 \ge \Delta/2$, vertices that are within squared distance
$\Delta/20$ in $X$ metric are also within squared distance $2\Delta/5$
in $Z$ metric (\Cref{cor:singlescalemapping}). By spreading
constraints there can only be at most $5n/4r$ such vertices.
\subsubsection{Obstacle Sets}
\label{subsec:obstacle}
The plan of the proof is to apply cover composition from \cite{ARV} in
order to find a short path (in graph metric) that crosses a lot of
cuts. At any step $i$, let $V_i$ be the remaining vertices in the
partition oracle. Let $Q_i$ be the set of vertices in $V_i$ that have
large $s$ values (at least $D_2/10$ as in \Cref{lem:singlescale}. In
this case there are two kinds of obstacle sets that prevents us from
applying the cover composition argument.

\begin{definition}[Obstacle Sets]
  At some step $i$ of the partition oracle, a set $S\subset V_i$ is an
  obstacle set of type I if it has size at least $n/F(r)$, and the
  $D_0$ neighbourhood contains at most $100n/r$ vertices in $V_i$.

  A set $T\subset Q_i$ is an obstacle set of type II, if it has size
  at least $10n/r$, and the $D_0$ neighbourhood contains at most
  $100n/r$ vertices in $V_i\backslash Q_i$.
\end{definition}

Using region-growing arguments, we can remove the obstacle sets using
the partition oracle without removing many vertices.




\begin{lemma}
  \label{lem:obstacleI}
  \label{cor:expansion}
  For a partition oracle with distance $D_0$ as in
  \Cref{lem:singlescale}, if at some step $i$, $\sum_{j\le i, |S_j'|
    \le 100n/r} |S_j'| \ge n\alpha/10\Delta\log 30r$, then one of the
  $S_j'$ of size at most $100n/r$ has expansion at most $\alpha$.
  
  In particular, there is a set $H \subset V_i$ whose size is at least
  $|V_i| - n/10\Delta\log 30r$, such that any subset $S\subset H$ of
  size at least $n/F(r)$ expands to at least $100n/r$ vertices in $H$.
\end{lemma}
\begin{proof}
  If we take the sum of capacity of all outgoing edges from these
  $S_j'$, each edge in the partition is counted at most twice,
  therefore
  $$\sum_{j\le i, |S_j'| \le 100n/r} |E(S_j',V\backslash S_j')| \le n\alpha/10\Delta\log 30r.$$
  On the other hand we know the sum of sizes is at least $n/10\Delta\log
  30r$, by averaging argument there must be one set that has expansion
  $\alpha$.
\end{proof}
For proving \Cref{lemma:dualsolution} this $\alpha$ will be
chosen as $O(d\sqrt{\log n}/\sqrt{\log r})$.

Notice that it is very important that the algorithm always uses the
partition oracle when it wants to remove a set of vertices. If the
algorithm simply removes a set of vertices, it will be hard to bound
the number of edges cut, and \Cref{cor:expansion} is no longer
true. In this case we may have many obstacle sets of type I and cannot
find a matching cover.

For obstacle sets of type II, since a large fraction of the vertices
in their neighbourhood have large $s_i$, they cannot cover a lot of
vertices without contradicting the validity of the dual solution

\begin{lemma}
\label{lem:obstacleII}
Use the partition oracle in \Cref{lem:singlescale} to remove
obstacle sets of type II. At any step, let $II$ be the set of steps
where a set of type II is removed. Then $\sum_{j\in II} |S_j'| \le
n/20\Delta\log 30r$.
\end{lemma}

\begin{proof}
  By the definition of obstacle sets of type II, we know each $S_j'$
  contains at least $10n/r$ vertices with $s$-value at least
  $D_2/10$. On the other hand, it contains at most $90n/r$ vertices
  with $s$-value smaller than $D_2/10$. Therefore $1/10$ fraction of
  the vertices in $S_j'$ have $s$ value at least $D_2/10$.
  $$\sum_{u\in V} s_u \ge \sum_{j\in II}\sum_{u\in S_j'} s_u \ge \frac{|S_j'|}{10} \cdot \frac{D_2}{10}.$$
  On the other hand $\sum_{u\in V} s_u \le \beta n$, so when $\beta =
  C\log^{-2} r$ for small enough $C$ we know $\sum_{j\in II} |S_j'|
  \le n/20\Delta\log 30r$.
\end{proof}

\subsubsection{Gaussian Projections and Matching Covers}
\label{subsec:matchingcover}
Recall the definitions of Matching Covers and Uniform Matching Covers
in~\cite{ARV}:
\begin{definition}
  A $(\sigma, \delta, c')$-matching cover of a set of points is a set
  $\mathcal{M}$ of matchings such that for at least a fraction
  $\delta$ of directions $u$, there exists a matching $M_u\in
  \mathcal{M}$ of at least $c'n$ pairs of points, such that each pair
  $(i,j)\in M_u$ are within graph distance $2D_0$, and
  satisfies $$\left<X_i-X_j, u\right> \ge
  2\sigma\sqrt{\Delta}/\sqrt{d}.$$
  The associated matching graph $M$ is defined as the multigraph
  consisting of the unions of all matchings $M_u$.
\end{definition}

\begin{definition}
  A set of matchings $\mathcal{M}$
  $(\sigma,\delta)$-uniform-matching-covers a set of points $S$ if for
  every unit vector $u$, there is a matching $M_u$ of $S$ such that
  every $(i,j)\in M_u$ is $2D_0$ close in graph distance, satisfies
  $|\left<u,X_i-X_j\right>| \ge 2\sigma\sqrt{\Delta}/\sqrt{d}$, and
  for every $i$, $\mu(u:i \mbox{ matched in } M_u) \ge \delta$.
\end{definition}

Notice that in addition to the properties in \cite{ARV}, we further
require that every matched pair must be close in graph distance.

Let the dimension of $X_i$'s be $d$. Let $u$ be a uniformly random
unit vector, when $d$ is large enough we know
\begin{lemma}
\label{lem:gaussian}
There exists thresholds $0 < \theta_1 < \theta_2$ such that $\theta_2
- \theta_1 = \Omega(\sqrt{\log r}/\sqrt{d})$ and polynomial
$G(r)$. Let $Good(u)$ be the event that number of vertices with
projection more than $\theta_1\Delta$ is smaller than $5n/r$.  For any
vertex $i$ whose $\norm{Z_i}^2 \ge \Delta/2$, $\Pr[u\cdot X_i\ge
\theta_2 \sqrt{\Delta} \mbox{ and }Good(u)] \ge 1/G(r)$.
\end{lemma}
\begin{proof} (sketch) We know each such $i$ has at most $5n/4r$
  closeby points. For points that are not close, conditioned on $i$
  has large projection, the probability that they also have pretty
  large projection is very small. Hence conditioned on $i$ being in,
  the expected number of vertices that have large projection is
  small. By Markov we know $\Pr[Good(u)|u\cdot X_i\ge
  \theta_2\sqrt{\Delta}] \ge 1/2$.
\end{proof}
Given a set of vertices $V_t$, which can be partitioned into three
parts $P,Q,R$, vertices $i\in P$ all have $\norm{Z_i}^2 \ge \Delta/2$,
vertices $j\in Q$ all have $s_j \le D_2/10$, and $|R| \le n/F(r)$
(notice that this is exactly what's guaranteed by
\Cref{lem:singlescale}), we will use the following algorithm to
find matching covers:
\framebox{
\parbox{0.95\textwidth}{
{\sc Construct Cover (P,Q,R)}
\begin{enumerate}\itemsep 0pt
\item Pick uniformly random unit vector $u$.
\item Let $Left = \{i:i\in P\mbox{ and } \left<X_i,u\right> \ge
  \theta_2\sqrt{\Delta}\}$, $Right = \{i:i\in P\mbox{ and }
  \left<X_i,u\right> \le \theta_1\sqrt{\Delta}\}$.
\item While exists pair $i\in Left$ and $j\in Right$ within graph
  distance $2D_0$.
\item $\qquad$ Match $(i,j)$, remove $i,j$ from $Left,Right$.
\item Fail if $|P\backslash Right| < 5n/r$ and number of unmatched
  vertices in $Left$ is at least $n/F(r)$
\end{enumerate}
}}

If the algorithm fails, the following Lemma shows that we will have an
obstacle set of type I or II.
\begin{lemma}
\label{lem:coverfail}
If Construct Cover fails, then
it finds an obstacle set of I or II.
\end{lemma}

\begin{proof}
  Let $S$ be the set of vertices that are left unmatched in
  $Left$. Let $\Gamma_{D_0}(S)$ be the $D_0$ neighbourhood of $S$ in
  $P\cup Q\cup R$. If $|\Gamma_{D_0}(S)| \le 100n/r$ first case of the
  Lemma is satisfied.

  If $|\Gamma_{D_0}(S)| > 100n/r$, then either $|\Gamma_{D_0}(S)\cap
  P| > 80n/r$, in which case by simple counting argument there must be
  a point left in $Right$ that is close to some point in $S$, and
  these two vertices can be matched (this contradicts with the
  assumption). When $|\Gamma_{D_0}(S)\cap P| \le 90n/r$, let $T =
  \Gamma_{D_0}(S)\cap Q$. Clearly $|T| > 10n/r$, and
  $\Gamma_{D_0}(T)\cap P \subset \Gamma_{2D_0}(S) \cap P$. The number
  of $2D_0$ neighbours of $S$ in $P$ cannot be more than $90n/r$
  (otherwise we will be able to find a matching pair), hence
  $|\Gamma_{D_0}(T)\cap P| \le 90n/r$ and $|\Gamma_{D_0}(T)\cap (P\cup
  R)| \le 90n/r+n/F(r) < 100n/r$.
\end{proof}

\cite{ARV} has a Lemma that shows matching covers imply uniform
matching covers. However in our situation, in order to apply the cover
composition Lemma, we need a really large uniform matching cover,
which is not guaranteed by the Lemma in \cite{ARV}.

If Construct Cover does not fail with polynomial probability, then
\Cref{lem:gaussian} means the matching cover is already
``almost'' uniform, in the sense that if we ignore the fact that
$n/F(r)$ points will not be matched, each vertex will be in the
matching with probability at least $1/G(r)$.

\begin{lemma}
\label{lem:bigcover}
If Construct Cover fails with probability less than $1/n^2G(r)$, then
there is a set $W\subset P$ of size at least $|P|-4|P|/r$ that is
$(1/G(r)r, \theta_2-\theta_1)$ uniformly matching covered.
\end{lemma}

\begin{proof}
  Consider the matching graph. First, even for the $n/F(r)$ points
  that remains unmatched, consider that they are matched to
  something. In this case each vertex has degree at least $1/G(r)$ by
  \Cref{lem:gaussian}.

  Now remove the edges that correspond to unmatched edges. In this
  step we have removed at most $|P|\cdot 1/G(r)r$ volume. Then we
  repeatedly remove any vertex that has degree at most
  $1/G(r)r$. Again we will remove at most $|P|/G(r)r$ volume. So the
  total volume removed is bounded by $2|P|/G(r)r$.

  However, we know that each vertex in $P$ started with degree at
  least $1/G(r)$. Removing $2|P|/G(r)r$ volume can reduce the degree
  of at most $4|P|/r$ vertices to below $1/2G(r)$. Therefore at most
  $4|P|/r$ vertices are removed.
\end{proof}

\subsubsection{Adapting ARV}
\label{subsec:ARV}

Using the uniform matching cover constructed above, and mechanisms in
\cite{ARV}, we can get the following Lemma.

\begin{lemma}
  \label{lem:ARV}
  If $W\subset V$ and has $(1/G(r)r, \Omega(\sqrt{\log r}/\sqrt{d})$
  uniform matching cover. Then there exists an algorithm that either
  finds $i,j\in W$, such that $d_{i,j} \le D_1 = O(D_0 \cdot \sqrt{\log
    n}/\sqrt{\log r})$, and $\|X_i-X_j\|^2 \ge D_2 = \Omega(\Delta/\log
  r)$, or finds a set whose $2D_0$ neighbourhood has size smaller than
  $100n/r$ in $P$.
\end{lemma}

\begin{proof}
  The proof follows from \cite{ARV}, the algorithm basically follows
  the cover composition proof, maintaining the cover $S_k$ along the
  induction steps (this is possible because the probabilities we are
  dealing with are all larger than some inverse polynomial, and the
  probabilities do not need to be estimated exactly). The main
  differences are:
  \begin{enumerate}
  \item Here we need to boost the probability from $1/G(r)r$ to
    $1-1/r$, this is $D_2 = \Omega(\Delta/\log r)$ (in \cite{ARV} we
    can find a pair that are constant distance away in $\ell_2^2$
    metric).
  \item The definition of non-expanding set is now a set that does not
    expand to $100n/r$ vertices within graph distance $2D_0$. This is
    OK because either there is a pair within graph distance $2D_0$ and
    $\ell_2^2$ distance more than $D_2$, in which case the Lemma is
    true; or all vertices in this neighbouring set are also close in
    $\ell_2^2$ distance, which then matches the definition of
    non-expanding set in \cite{ARV}. \qedhere
  \end{enumerate}
\end{proof}

\subsubsection{Final Proof}
\label{subsec:final}
The following Lemma immediately implies \Cref{lem:weaksseflowexists}.
\begin{lemma}
  \label{lemma:dualsolution}
  Given a dual solution with degree $d$ and expansion $\beta < C_\beta
  \log^{-2} r$ (where $C_\beta$ is a universal constant), there is an
  algorithm that finds a set of size at most $100n/r$ with expansion
  $O(d\sqrt{\log n}/\sqrt{\log r})$.
\end{lemma}
\begin{proof}
  First apply \Cref{lem:singlescale}.  If
  \Cref{lem:singlescale} did not find a set, then we have sets
  $P,Q,R$ from \Cref{lem:singlescale} and a partition oracle
  whose current set is $P\cup Q \cup R$.

  Now we shall repeatedly apply Construct Cover. In this case we can
  get an obstacle set of type I or II.

  If it fails with more than $1/n^2G(r)$ probability then we get an
  obstacle set from \Cref{lem:coverfail}. Otherwise we would have
  a large uniform matching cover by \Cref{lem:bigcover}. Then we
  apply \Cref{lem:ARV} on this uniform matching cover, under the
  assumptions $D_1 < D_2/10$ \footnote{Notice that the constant in
    $D_1$ is in fact hiding in the expansion $O(d\sqrt{\log
      n}/\sqrt{\log r})$ which can be chosen independently of $D_2$.},
  since $W\subset P$ the first case of \Cref{lem:ARV} cannot
  happen. We must get a non-expanding set. \Cref{lem:coverfail}
  also applies to this non-expanding set and we can again get an
  obstacle set of type I or II.

  Once we get the obstacle set, feed that set into the partition
  oracle, and recurse on the current set of the oracle. We always call
  obstacle sets of type I $S$, and obstacle sets of type II $T$. The
  corresponding sets returned by the oracle will be called $S'$ and
  $T'$, respectively.

  At the end one of the two cases will happen: Either the sets
  corresponding to $S'$ take up more than $n/10\Delta\log 30r$
  vertices or the sets corresponding to $T'$ take up more than
  $n/10\Delta\log 30r$ vertices.

  In the first case \Cref{cor:expansion} shows one of the $S'$
  must have low expansion.

  The second case contradicts the feasibility of dual solution because
  of \Cref{lem:obstacleII}.
  %
\end{proof}

\subsection{Getting SSE Flows and Spectral SSE Flows}
Thus far our existence proof dealt with weak SSE flows.
\subsubsection{Getting to SSE flows}
\begin{lemma}
  \label{lem:weaktosseflow} If $G = (V,E)$ is a graph and a $(r, d,
  \beta)$ weak SSE flow is embeddable in $G$, either there is a set of
  size at most $n/r$ that has expansion less than $d\beta$, or there
  exists a $(r, d, \beta/6)$ SSE flow embeddable in $G$.
\end{lemma}
\begin{proof}
  Let $F$ be the weak SSE flow. If for all sets $S$ of size $|S| \le
  n/3r$, the $F$ has expansion at least $\beta$, then $F$ is already a
  SSE flow.

  When there exists $S$ of size smaller than $n/3r$ and the expansion
  in $F$ is smaller than $\beta$, remove $S$ (for remaining vertices
  replace edges going to $S$ with self-loops) and repeat this
  procedure.

  If the union of the removed sets is $U$, the size of $U$ cannot be
  larger than $n/3r$: if after removing some $S$ the size of $U$ first
  become larger than $n/3r$, then since $S$ has size smaller than
  $n/3r$, the size of $U$ must be between $n/3r$ and $2n/3r$. The
  expansion of $U$ is at most the maximum expansion among sets $S$,
  which is smaller than $\beta$. Such a set cannot exist by the
  definition of weak SSE flows.

  Now add a source and a sink to the graph. Add an edge from source to
  every vertex in $U$ with capacity $d\beta$, add an edge from every
  vertex in $V\backslash U$ to the sink with capacity $d\beta$, and
  then try to route the maximum single-commodity flow from source to
  sink.

  If the maximum flow is smaller than $d\beta|U|$, then there must be
  a cut of value smaller than $d\beta|U|$ in the new graph. Let $Q$ be
  one side of this cut that contains the source, then $E(Q,V\backslash
  Q) < d\beta|U| - |Q\oplus U|d\beta$ (this is because, for every $i$
  in $Q$ but not $U$, it has degree $d\beta$ to the sink; for every
  $i$ in $U$ but not $Q$, it has degree $d\beta$ from the source), and
  $|Q| \ge |U| - |Q\oplus U|$. The expansion of $Q$ is strictly
  smaller than $d\beta$.

  If the maximum flow is $d\beta|U|$, let the single-commodity flow be
  $F_1$, and let $F_2 = (F+F_1)/2$ (here ``$+$'' just take the linear
  combination of demands). Clearly $F_2$ is still embeddable into
  $G$. For any set $S$ of size at most $n/r$, if more than $|S|/3$ of
  the vertices are outside $U$, then it already has $\beta d|S|/6$
  outgoing edges outside $U$ in $F/2$; if less than $|S|/3$ of the
  vertices are outside $U$, then it has $d\beta|S|/6$ outgoing edges
  just by the flow $F_1/2$. Therefore $F_2$ is a $(r,d,\beta/6)$ SSE
  flow.
\end{proof}

Unfortunately, this Lemma is only existential. In general, even if we
are given a SSE flow, it is hard to verify it exactly.

\subsubsection{Getting Spectral SSE flow}
We can use higher order equivalents of Cheeger's
Inequality to establish a relation between SSE flows and spectral SSE flows:
\begin{theorem}[\cite{lrtb12,lgt12}]
  \label{thm:highcheeger}
  For any graph $G$,
  $\Phi_r \le O(\sqrt{\lambda_{2r}(\mathcal{L})\log r}).$
  Here $\lambda_{2r}(\mathcal{L})$ is the $2r$-th smallest eigenvalue
  of the normalized Laplacian of $G$.
\end{theorem}

This implies that if the largest and smallest degree are close, then
an SSE flow is already a spectral SSE flow.

\begin{lemma}
  \label{lem:combtospectral}
  For any graph $G=(V,E)$, if there is a $(r,d,\beta)$ SSE-flow
  embeddable in $G$, then there is a $(2r,d, \Omega(\beta^2/\log r))$
  spectral SSE flow embeddable in $G$.
\end{lemma}
Before proving~\Cref{lem:combtospectral}, we will need the following
simple claim so as to relate the eigenvalues of {\em normalized}
Laplacian matrix to the original Laplacian.
\begin{clm} \label{thm:laplacian-norm-to-unnorm} Let
  $d_{\min}$,$d_{\max}$ be the minimum and maximum degrees in $G$,
  respectively. Then:
  $$
  \frac{1}{d_{\max}} L(G) \preceq \mathcal{L}(G) \preceq
  \frac{1}{d_{\min}} L(G).
  $$
\end{clm}
\begin{proof}
  For any pair of nodes $u,v$, $d_{\min} \le \sqrt{d_u d_v} \le
  d_{\max}$. Hence for any $x\in \R^V$:
  \[\frac{x^T L(G) x}{d_{\max}} =
  \sum_{u<v} \frac{C_{uv}}{d_{\max}} \left(x_u - x_v\right)^2 \le x^T
  \mathcal{L} x = \sum_{u<v} \frac{C_{uv}}{\sqrt{d_u d_v}} \left(x_u -
    x_v\right)^2 \le \frac{x^T L(G) x}{d_{\min}}. \tag*{\qedhere}
  \] 
\end{proof}
\begin{proof}[Proof of~\Cref{lem:combtospectral}]
  Let $F$ be the $(r,d,\beta)$ SSE-flow, let $F_1$ be a flow whose
  demands are $\delta_{i,j} = c_{i,j}$. Clearly $F_1$ is embeddable in
  $G$ and has degree $1$. Let $F_2 = F/2+dF_1/2$, then the degrees of
  vertices in $F_2$ are between $d/2$ and $d$.

  By definition of SSE flow we know $\Phi_r(F_2) \ge \beta/2$. Let
  $\mathcal{L}$ be the normalized Laplacian of $F_2$, and $L$ be its
  Laplacian, then by \Cref{thm:highcheeger}
  $\lambda_{2r}(\mathcal{L}) \ge \Omega(\beta^2/\log r)$.

  Since the degrees of $F_2$ are all between $d/2$ and $d$,
  by~\Cref{thm:laplacian-norm-to-unnorm}, the eigenvalues of its
  normalized Laplacian are closely related to its Laplacian:
  $\lambda_{2r}(L) \ge \frac{d}{2} \lambda_{2r}(\mathcal{L}) =
  \Omega(d\beta^2/\log r)$.
\end{proof}

The inverse direction (spectral flows imply combinatorial flows) is
also true, except the combinatorial expansion must be defined on $r$
disjoint sets instead of one set.

\begin{lemma}[\cite{KLLGT}]
\label{lem:spectraltocomb}
A $(r,d,\lambda)$ spectral flow satisfies the following combinatorial
expansion property: for any $r$ disjoint sets $S_1, S_2,..., S_r$, the
maximum of the expansion of these sets is at least $\lambda/2$.
\end{lemma}


\begin{proof}
  This proof comes from \cite{KLLGT}, we restate it here for
  completeness.
  We use Courant-Fischer-Weyl characterization the variational
  definition of $r^{th}$ smallest eigenvalue:
  $$\lambda_r(L) = \min_{\mbox{\small subspace $P$ of dimension $r$}}
  \max_{h\in P} \frac{h^T L h}{h^T h}.$$
  Let the subspace $P$ be the span of the indicator vectors of
  $S_i$'s. For any $h = \sum_{i=1}^r \lambda_i \vec{1}_{S_i}$, for all
  $u,v\in V$,
  $$(h(u)-h(v))^2 \le \sum_{i=1}^r 2\lambda_i^2 (\vec{1}_{S_i}(u) - \vec{1}_{S_i}(v)^2$$
  So the Rayleigh Quotient of $h$ is at most
  $$
  R(h) = \max_{h\in P} h^T L h/\norm{h}^2 \le \frac{2\sum_{i=1}^r
    \lambda_i (\vec{1}_{S_i}(u) - \vec{1}_{S_i}(v)^2}{\sum_{i=1}^k
    \lambda_i^2\norm{\vec{1}_{S_i}}^2} \le 2\max_{i\in[r]}
  R(\vec{1}_{S_i}).$$
  We know $\max R(h) \ge \lambda$, so the maximum expansion must be at
  least $\lambda/2$.
\end{proof}
In order to find spectral SSE flows, the following algorithm uses a
convex program:
\begin{lemma}
\label{lem:findspectralsseflow}
If there exists a $(r,d,\lambda)$ spectral SSE flow embeddable in $G$,
there is an efficient algorithm that finds a $(2r, d, \lambda/2)$
spectral SSE flow.
\end{lemma}
\begin{proof}
  The algorithm tries to solve the following optimization problem:
  \begin{align*}
    \max & \sum_{i=1}^{2r} \lambda_i(L(F)) \\
    s.t. \forall i\in V \quad & \frac{d}{2} \le \sum_{j\in V} \sum_{p\in\mathcal{P}_{i,j}} f_p \le d. \\
    & F\mbox{ embeddable in }G.
  \end{align*}
  Here $L(F)$ is the Laplacian of the flow. The first constraint just
  says the degree of every vertex should be between $d/2$ and
  $d$. This is a convex program because entries of $L(F)$ are linear
  functions over $f_p$, and the sum of first $2r$ eigenvalues of a
  matrix is a concave function. The convex program can be solved in
  polynomial time.\footnote{There are exponentially many paths, but
    there is a canonical way of reducing the number of variables for
    flows.}
  Clearly the $(r,d,\lambda)$ spectral SSE flow is a feasible solution
  and has objective value at least $rd\lambda$. Hence the solution of
  this convex program must have objective function at least $rd\lambda$,
  which means the $2r$-th eigenvalue of $L(F)$ is at least
  $\frac{rd\lambda}{2r} = d\lambda/2$.
\end{proof}
\subsection{Finding a Small Nonexpanding Set when Eigenspace
  Enumeration Fails}
Combining \Cref{thm:spectralSSEexist,thm:SSE2}, we know if
$\Phi_{local}$ for a graph is at least $O(\Phi_{global}\sqrt{\log n}
\log ^{4.5} r/\epsilon)$, there is an eigenspace enumeration algorithm
that finds a $(1+\epsilon)$ approximation of sparsest cut. Here we
show when the algorithm fails, how to find a small set that does not
expand in polynomial time.
\begin{lemma}
  \label{lem:finddual}
  Given a graph $G$, for any $d$, $r$, there is a polynomial time
  algorithm that either finds a $(2r,d,\lambda = \Omega((\log
  r)^{-5}))$ spectral flow, or finds a set of size at most $100n/r$
  that has expansion at most $O(d\sqrt{\log n}/\sqrt{\log r})$.
\end{lemma}
\begin{proof}
  By \Cref{lem:weaktosseflow,lem:combtospectral}, we know a weak SSE
  flow implies a spectral SSE flow unless there is a small set with
  very small expansion.  Therefore if the algorithm in
  \Cref{lem:findspectralsseflow} does not work, either there is a set
  of size at most $n/3r$ that has expansion $d \beta$ where $\beta =
  \theta((\log r)^{-2})$, or there is no weak SSE flow.

  In the first case we can simply run the approximation algorithm for
  small set expansion in \cite{bansalsse11}, which gives a $\sqrt{\log
    n \log r}$ approximation, the set we get will be small and has
  expansion at most $d\beta \sqrt{\log n \log r} < O(d\sqrt{\log
    n}/{\log r})$

  In the second case, there is no weak SSE flow, so the LP for the
  weak SSE flow must be infeasible, and its dual must be feasible. The
  original dual formulation has exponentially many variables, however
  in \Cref{subsec:mapping} we mapped the solution to a concise
  representation using $l_2^2$ vectors $Z_i$'s.
  \Crefrange{eqn:newdualvalue}{eqn:newpath} are almost constraints of
  a semidefinite program, except for \cref{eqn:spreading}. However, we
  can write the spreading constraint in more tractable way:
  $$
  \forall i\in V \quad \sum_{j\in V} \min \{0.9\norm{Z_i}^2,
  \norm{Z_j-Z_i}^2\} \ge 0.9\norm{Z_i}^2 n - \norm{Z_i}^2 \cdot
  \frac{n}{r}.
  $$
  This equation is clearly satisfied by the $Z_i$'s converted from the
  original dual solution, because there the number of vectors within
  $0.9\norm{Z_i}^2$ is at most $10n/9r$, even if all of them are
  identical with $Z_i$, the sum on the LHS can only be $\norm{Z_i}^2
  \cdot \frac{n}{r}$ away from its maximum possible value
  $0.9\norm{Z_i}^2 n$.

  On the other hand, if this equation is satisfied, we know for any
  $i$, the number of $j$ such that $\norm{Z_j-Z_i}^2 \le
  0.8\norm{Z_i}^2$ is at most $10n/r$. This is very similar to
  Constraint (\ref{eqn:spreading}) except the constants are
  larger. This increase in constant does not change anything in the
  proof of \Cref{lemma:dualsolution}.

  Therefore, we can solve the SDP to get a concise representation of
  the dual solution, and then apply \Cref{lemma:dualsolution} to find
  a set that has size at most $100n/r$ with expansion $O(d\sqrt{\log
    n}/\sqrt{\log r})$.
\end{proof}


\section{Planted Expander Model}
\label{sec:semirandom}
Our algorithm naturally applies to the {\em planted expander
  model}. In this model the graph has a planted bisection of expansion
$\Phi_{planted}$, the smaller side of the bisection has size $\rho
n$. The induced subgraph on each side of the partition is an expander
with expansion $\Phi_{global} \gg \Phi_{planted} \sqrt{\log n \log 1/\rho}$. In
this case we can show the assumptions in \Cref{thm:main1} hold, and
the algorithm gives a good approximation to sparsest cut.

This result is similar to the ``planted spectral expander model'' in
\cite{mmv}. The main difference is that they assume the induced graphs
of the partition have {\em algebraic} expansion constant times more
than $\Phi_{planted}$. Notice that our combinatorial expansion
property only implies algebraic expansion of $\Phi_{planted}^2 \log n
\log 1/\rho$, which might be smaller than $\Phi_{planted}$ if the
planted bisection is sparse enough. Unfortunately our result only
applies to regular graphs, therefore a comparison is not possible {\it
  per se}.
%
%
Our formal guarantee is given in the following theorem.
\begin{theorem}
  \label{thm:semirandom}
  Assume graph $G = (V,E)$ is a regular graph with an unknown planted
  bisection $(S, V\backslash S)$. The size of $S$ is $\rho n$ ($\rho
  \le 1/2$) with $\Phi(S)=\Phi_{planted}$. If the induced subgraphs of
  $S$ and $V\backslash S$ both have expansion $\Phi \gg
  \frac{1}{\epsilon^{1.5}} \Phi(S)\sqrt{\log n \log \frac{1}{\rho
      \epsilon}}$, then the algorithm in \Cref{thm:main1} with $r =
  O(1/\rho)$ gives a $(1+\epsilon)$ approximation to sparsest cut.
\end{theorem}
\begin{proof}
  We only need to show the assumptions in \Cref{thm:main1} are
  satisfied: Sets of size $\rho n/2$ should have sparsity at least
  $\Omega(\phi_{sparsest}\sqrt{\log n \log \frac{1}{\rho \epsilon}}
  \epsilon^{-1.5})$. Since sparsity and expansion are within a
  constant factor, we will show the expansion of small sets are at
  least $\Delta\triangleq \Omega(\Phi(S)\sqrt{\log n \log\frac{1}{\rho \epsilon}
  }/\epsilon^{1.5})$.

  For any set $T$ of size at most $\rho n/2$, let $T_1$ be $T\cap S$
  and $T_2$ be $T\cap (V\backslash S)$. By assumption we know
  $E(T_1,S\backslash T_1) \gg |T_1|\cdot \Delta
  $ and
  $E(T_2,V\backslash(S\cup T_2))\gg |T_2| \Delta
  $. Hence
  \[
  \Phi(T) = \frac{E(T,V\backslash T)}{|T|} \ge \frac{E(T_1,S\backslash
    T_1) + E(T_2,V\backslash(S\cup T_2))}{|T_1|+|T_2|} \gg
  \Delta. 
  \tag*{\qedhere}
  \]
\end{proof}


\end{document}